\documentclass[onecolumn]{IEEEtran}
\usepackage{amsmath}
\usepackage{amsfonts}%
\usepackage{amssymb}
\usepackage{wasysym}
\usepackage{times}
\usepackage{tabularx}
\usepackage{multirow}
\usepackage{enumerate}
\usepackage{subfigure}
\usepackage{bm}
\usepackage{bbm} 
\usepackage[dvipsnames,table]{xcolor}
\usepackage{arydshln}
\usepackage{url}
\usepackage{wrapfig}
\usepackage{sidecap}
\usepackage{newproof}
\usepackage{eurosym}
\usepackage{mathrsfs}
\usepackage{algorithmic}
\usepackage{algorithm}
\usepackage{stmaryrd}
\usepackage{upgreek}
\usepackage{tikz}
\usepackage{cite}
\usepackage{circuitikz}
\usetikzlibrary{arrows,decorations.pathreplacing}
\usetikzlibrary{matrix, calc,backgrounds}
\usepackage{authblk}
\usepackage{graphicx}

\newtheorem{definition} {Definition}
\newtheorem{te}{Theorem}

\newtheorem{ex}{Example}

\newcommand{\minf}{\scriptstyle{\text{-}\infty\;}}

\usepackage{soul}
\setstcolor{red} 
\setul{0.3ex}{0.7pt} 
\definecolor{darkgray}{rgb}{0.66, 0.66, 0.66}
\setulcolor{darkgray}

\usepackage{cancel}

\def\mathcolor#1#{\@mathcolor{#1}}
\def\@mathcolor#1#2#3{%
  \protect\leavevmode
  \begingroup
    \color#1{#2}#3%
  \endgroup
}

\usepackage{paralist} 
\usepackage{braket}
\usepackage{physics}

\usepackage{mathtools}

\usepackage{booktabs}  
\usepackage{array}     

\begin{document}
\title{On the Minimum Distances of Finite-Length \\Lifted Product Quantum LDPC Codes}
\author{Nithin Raveendran\IEEEauthorrefmark{1}, David Declercq\IEEEauthorrefmark{1}, and Bane Vasi\'{c}\IEEEauthorrefmark{2}\\
Department of Electrical and Computer Engineering\\ 
The University of Arizona, Tucson, AZ 85721\\ \IEEEauthorrefmark{1}\{nithin,daviddeclercq\}@arizona.edu, \IEEEauthorrefmark{2}vasic@ece.arizona.edu}
\maketitle

\begin{abstract}
Quantum error correction (QEC) is critical for practical realization of fault-tolerant quantum computing, and recently proposed families of quantum low-density parity-check (QLDPC) code are prime candidates for advanced QEC hardware architectures and implementations.
This paper focuses on the finite-length QLDPC code design criteria, specifically aimed at constructing degenerate quasi-cyclic symmetric lifted-product (LP-QLDPC) codes. We describe the necessary conditions such that the designed LP-QLDPC codes are guaranteed to have a minimum distance strictly greater than the minimum weight stabilizer generators, ensuring superior error correction performance on quantum channels. The focus is on LP-QLDPC codes built from quasi-cyclic base codes belonging to the class of type-I protographs, and the necessary constraints are efficiently expressed in terms of the row and column indices of the base code. Specifically, we characterize the combinatorial constraints on the classical quasi-cyclic base matrices that guarantee construction of degenerate LP-QLDPC codes. Minimal examples and illustrations are provided to demonstrate the usefulness and effectiveness of the code construction approach. The row and column partition constraints derived in the paper simplify the design of degenerate LP-QLDPC codes and can be incorporated into existing classical and quantum code design approaches.
\end{abstract}

\section{Introduction}
\emph{Sparse stabilizer codes} -- referred to generally as quantum low-density parity-check (QLDPC) codes -- have recently risen to the forefront of quantum error correction (QEC) research, both in terms of theoretical approaches and practical implementation aspects\cite{panteleev_asymptotically_2021,  nature_qldpc_demonstration, bravyi2023_IBM_highthreshold}. 
The shift in focus to implementing these non-local codes enabled to devise novel strategies for utilizing QLDPC codes in the hardware architectures\cite{bravyi2023_IBM_highthreshold,IBM_2024_LinearAncillaLogicalMeasurementsBBCode,IBM_2024_Cross_improvedqldpcsurgerylogical,xu2024_Fast_Parallel_LogicalComputation_Homological_Product,bluvstein2023Reconfig_Atom_Exp,Matching_GB_NeutralAtoms_2024}. By allowing the qubits to be connected in multiple layers and reshuffling them, stabilizer measurements that require physically long-range connections among qubits can be realized.
However, most of the theoretical focus on QLDPC code construction has mainly been on asymptotic analysis, particularly with the emphasis on the linear minimum distance scaling and good expansion properties \cite{dinur2023good,gu2024singleshotQLDPC}.
At the same time, hardware implementation demands shorter finite-length codes with properties such as high minimum distance, high code rate, and ease of hardware implementation of the syndrome circuitry. 
The overarching goal of this work is to provide a systematic way to design practical finite-length QLDPC codes without sacrificing the quantum minimum distance. 
In particular, our code construction methodology pertains to an important class of QLDPC codes obtained starting from classical LDPC codes using the recent Kronecker product-based constructions, most notably \emph{Lifted Product} (LP) codes. 
These LP-QLDPC codes proposed and designed by Panteleev and Kalachev had excellent code properties both in terms of minimum distance and code rate and promise for implementation~\cite{panteleev_degenerate_2021,panteleev_asymptotically_2021}.
This is an important subclass of Calderbank-Shor-Steane (CSS) stabilizer codes 
\cite{calderbank1996quantum_exists,Gottesman97}, because many quantum code families may be viewed as special cases or modifications of the LP-QLDPC codes. Furthermore,  these codes are easy to design from existing classical codes and demonstrate good error correction performance if one can exploit the symmetry and structure of LP-QLDPC codes to the benefit of iterative decoding as shown in \cite{Brest_2023_NN_Quantum_Parallel}. 

One of the drawbacks of existing LP-QLDPC constructions is the use of two sequential and independent steps: the choice of the classical base-LDPC code and the Kronecker product followed by the lifting step. 
However, these two steps of the code design should not be separated, and choosing a good base-LDPC code does not necessarily ensure that the resulting LP-QLDPC code has the desired properties, especially in terms of minimum distance. 
While many choices of parity check matrices of such codes show excellent decoding performance with a low probability of logical error after decoding, the original description of LP-QLDPC codes is generic and does not guarantee minimum distance for the constructed quantum codes.
For instance, using a single base matrix $B_1$ = $B_2 = B$ based on a classical code $\mathcal{C}$ with distance $d_{\text{min}}^\mathcal{C}$, we obtain a symmetric LP-QLDPC code $\mathcal{Q}$ $[[N,K,d_{\text{min}}^\mathcal{Q}]]$ having $d_{\text{min}}^\mathcal{Q}$. However, the LP-QLDPC code design does not guarantee that $d_{\text{min}}^\mathcal{Q} = d_{\text{min}}^\mathcal{C}$. 

In this paper, by properly choosing the base-LDPC matrix entries, we can give certain guarantees for the quantum minimum distance such that the LP-QLDPC code is degenerate - minimum distance exceeds the weight of stabilizer generators. When this condition holds, the code is said to be \textit{degenerate}. We obtain constraints wherein the LP construction is guaranteed to have logical codewords of lower Hamming weights, especially those with weight equal to the stabilizer weight, which is highly detrimental.
These low-weight codewords are the dominant cause of undetectable errors. 
A typical characteristic behavior among these spuriously appearing codewords is that they span the entire support of the LP-QLDPC code, similar to the stabilizers. 
We exploit this behavior to analytically evaluate the constraints on the base matrix and avoid such instances in our code design. The base codes chosen are quasi-cyclic LDPC codes from circulant permutation matrices (CPMs).  
Though quasi-cyclic LDPC code constructions are well studied in classical error correction, our effort marks the first for characterizing combinatorial constraints for the LP-QLDPC codes.

The paper is organized as follows: Section \ref{sec_Prelims} sets the preliminaries for understanding quasi-cyclic LDPC codes, stabilizer codes, and their logical operators; in Section \ref{sec_QLDPC_Code}, we describe the LP-QLDPC construction used in this work and their base matrix structure. In Section \ref{sec:MinDistLPCodes}, we describe the constraints on the classical base matrix required to ensure that the constructed LP-QLDPC code is degenerate. Section \ref{sec_CodeDesign} illustrates the advantages of our approach and Section \ref{sec:Conclusion} concludes the paper with a discussion on future research directions.

\section{Preliminaries}
\label{sec_Prelims}

\subsection{Quasi-Cyclic LDPC Codes}
\label{subsec_QC_LDPC_Codes}
Quasi-cyclic (QC) LDPC codes have been extensively studied in classical error correction and are incorporated into contemporary communication standards, such as IEEE 802.16e, IEEE 802.11n, and 5G~\cite{Fossorier04, Divsalar06, Smarandache12}. The quasi-cyclic structure of these codes facilitates low-complexity encoding and decoding processes, making them highly suitable for hardware implementation. 
The parity check matrix of a quasi-cyclic binary LDPC is formed of $L \times L$ blocks where each block is either an all-zero matrix or a circulant permutation matrix (CPM) with $L$ as the circulant size. 
A CPM is a structured matrix representing a permutation derived from a cyclic group of order $L$. 
More specifically, a CPM is defined as $\alpha^i$, the power of a primitive element of a cyclic group, where $i \in \{0,\ldots, L-1 \}$. 
$\alpha^i$ represents an $L \times L$ identity matrix $\mathbf{I}_L$ whose rows are each circularly shifted by $i$ positions to the left.
By convention, to represent the all-zero $L \times L$ matrix, we use $\alpha^{-\infty}=0$. 
The first column in $\alpha^i$ has a one in the $i$-th position and all other entries being zero, the second column has a one in the $ (i+1)$-th position, and so on, wrapping around cyclically to obtain the $L \times L$ 
 matrix. For $L = 3$, we have
$$\alpha^0 = \begin{bmatrix}
1 & 0 & 0 \\
0 & 1 & 0 \\
0 & 0 & 1
\end{bmatrix}, \quad
\alpha^1 = \begin{bmatrix}
0 & 0 & 1 \\
1 & 0 & 0 \\
0 & 1 & 0
\end{bmatrix}, \quad
\alpha^2 = \begin{bmatrix}
0 & 1 & 0 \\
0 & 0 & 1 \\
1 & 0 & 0
\end{bmatrix}, \quad
\alpha^{-\infty} = \begin{bmatrix}
0 & 0 & 0 \\
0 & 0 & 0 \\
0 & 0 & 0
\end{bmatrix}.
$$
One can express the PCM using the base matrix $\mathbf{B}$ in CPM form as in left of Eq. \eqref{eq:B_and_H} or a more convenient way using only exponents as the matrix entries $b_{i,j} \in \{\minf, 0, 1, \ldots, L-1\}$ in the right of \eqref{eq:B_and_H}. The corresponding integers represent the CPM values and $\minf$ is for the all-zero block matrix if any.

\begin{equation}
\mathbf{B} = \begin{bmatrix}
\alpha^{b_{1,1}} & \alpha^{b_{1,2}} & \cdots & \alpha^{b_{1,n}} \\
\alpha^{b_{2,1}} & \alpha^{b_{2,2}} & \cdots & \alpha^{b_{2,n}} \\
\vdots & \vdots & \ddots & \vdots \\
\alpha^{b_{m,1}} & \alpha^{b_{m,2}} & \cdots & \alpha^{b_{m,n}}
\end{bmatrix}_{m \times n}, \quad
B = \begin{bmatrix}
b_{1,1} & b_{1,2} & \cdots & b_{1,n} \\
b_{2,1} & \ddots & \ddots & \vdots \\
\vdots & \ddots & \ddots & \vdots \\
b_{m,1} & \cdots & \cdots & b_{m,n}
\end{bmatrix}_{m \times n}.
\label{eq:B_and_H}
\end{equation}

\emph{Conjugate Transpose of Base Matrix}: Define $\mathbf{B}^*$ as the conjugate transpose of the base matrix $\mathbf{B}$ by replacing the corresponding entries $\alpha^{b_{i,j}}$ with the conjugate values $\alpha^{(L- b_{i,j})\mod L}$, in the transposed base matrix. The minus $-$ operation is performed modulo-$L$ and for the $b_{i,j} = -\infty$ is defined to remain unchanged as $L - (-\infty) = -\infty.$ Hence, 
\begin{equation}
B^* = \begin{bmatrix}
L - b_{1,1} & L - b_{2,1} & \cdots & L - b_{m,1} \\
L - b_{1,2} & \ddots & \ddots & \vdots \\
\vdots & \ddots & \ddots & \vdots \\
L - b_{1,n} & \cdots & \cdots & L - b_{n,m}
\end{bmatrix}_{n \times m}.
\end{equation}
The process of obtaining a parity-check matrix from such a base matrix is called \emph{lifting}, denoted by the mapping $\mathbbm{L}(.)$ from the base matrix $\mathbf{B}$ to the binary PCM. After lifting, the binary matrix $\mathbf{H} = \mathbbm{L}(\mathbf{B})$ has $L \times m$ rows and $L \times n$ columns. 

\begin{ex}
Let us look at a parity check base matrix of a QC-LDPC code (same code as provided in Example 2 in \cite{Tanner_QC}) with circulant size $L=7$ and base matrix 
\begin{equation}
    \label{Eq:Example1BaseMX}
B=\begin{bmatrix}
1&2&4\\
6&5&3
\end{bmatrix}.
\end{equation}
Upon lifting, we obtain the PCM of the associated code $\mathbf{H} = \mathbbm{L}(\mathbf{B})$ with $14$ rows and $21$ columns. The QC-LDPC code $\mathcal{C}: [21,8,6]$ encodes $8$ bits into $21$ bits with a minimum distance $d_\text{min}^{\mathcal{C}} = 6$. 
\end{ex}

In our analysis of the minimum distances of LP-QLDPC codes, we look closely at base matrices whose CPM values are all integers as in the example above. Since every column and row contains $m$ and $n$ integers, these codes belong to the fully connected $(m,n)$ regular QC-LDPC codes. The above example is for a (2,3) regular QC-LDPC code. More generally, we can build LP-QLDPC codes from protograph-based codes (of various types) with multiple edges - bringing more flexibility to the code design. We discuss mainly type-1 protograph codes \cite{arrayCode_fan} - which allows for up to weight-1 circulants. 
These type-1 protographs are popular due to their simplicity and straightforward structure, which makes the code easier to analyze and implement. In \cite{MD99}, the upper bound on the minimum distance of such codes is obtained as $(m + 1)!$, followed by the necessary condition to reach this bound in \cite{Fossorier04}. These bounds are generalized (Theorems 7 and 8 in \cite{Smarandache12}) for the case of other QC LDPC codes. 
In this work, we will explore the type-1 protograph-based LP-QLDPC code constructions and leave the more general protographs for future work.

\subsection{Quantum Stabilizer Codes}
An $[[N,K]]$ quantum stabilizer code \cite{Gottesman97} is a $2^K$-dimensional subspace $\mathcal{Q}$ of the Hilbert space $(\mathbb{C}^2)^{\otimes N}$ given by the common $+1$ eigenspace of the stabilizer group $\mathcal{S}$, i.e.,
\begin{equation} \label{eq:StabilizerCode} \mathcal{Q} = \{\ket{\psi} \in \mathbb{C}^{2^N} \colon \mathbf{S}_i\ket{\psi}= \ket{\psi}, \forall i, \mathbf{S}_i \in \mathcal{S} \}. \end{equation}
For \emph{Calderbank-Shor-Steane} (CSS) stabilizer codes, the stabilizer generators $\mathbf{S}_i, \forall i \in [M]$ are the $N$-fold tensor product of the operators either from the subset $\{\mathbf{I}, \mathbf{X}\}$ or from the subset $\{\mathbf{I},\mathbf{Z}\}$ of the set of Pauli matrices $\{\mathbf{I}, \mathbf{X}, \mathbf{Z}, \mathbf{Y}\}$~\cite{calderbank1996quantum_exists,Steane-physreva96}. In this work, we focus on the constructions of quantum LDPC codes that fall under the CSS category.  

Every (unsigned) Pauli operator can be mapped to a binary tuple as follows: $\mathbf{I} \mapsto (0,0),\; \mathbf{X} \mapsto (1,0),\;\mathbf{Z} \mapsto (0,1),\;\mathbf{Y} \mapsto (1,1)$.
Thus, Pauli operators on $N$ qubits have a length $2N$ binary vector representation, enabling us to map every element of the stabilizer group $\mathcal{S}$  (modulo global phases) into its corresponding binary form. 
This mapping gives a matrix representation of the CSS code stabilizer generators called the parity-check matrix (PCM) $\mathbf{H}$, which is given by 
\begin{equation}
\mathbf{H} = 
\begin{bmatrix}
\mathbf{H}_{\mathrm{X}} & \mathbf{0}\\
\mathbf{0} & \mathbf{H}_{\mathrm{Z}}
\end{bmatrix},
\end{equation}
where $\mathbf{H}_{\mathrm{X}}$ and $\mathbf{H}_{\mathrm{Z}}$ represent  $M_\mathrm{X} \times N$ and  $M_\mathrm{Z} \times N $ binary matrices to correct phase-flip, $\mathrm{Z}$ and bit-flip, $\mathrm{X}$ errors, respectively. The commutativity property that stabilizers satisfy in Pauli notation is equivalent to the symplectic inner product in the binary representation \cite{Nielsen}. Therefore, using the transpose of a matrix $\mathbf{A}$ denoted by $\mathbf{A}^{\mathsf{T}}$, the PCMs of the CSS codes should satisfy  
\begin{equation}
\mathbf{H}_{\mathrm{X}}\left(\mathbf{H}_{\mathrm{Z}}\right)^{\mathsf{T}}=\mathbf{0}. 
\label{eq:OrthogonalPCs}
\end{equation}

A generating set of the non-trivial logical Pauli operators - those that are orthogonal to both $\mathbf{H}_{\mathrm{X}}$ and $\mathbf{H}_{\mathrm{Z}}$ parity check rows, can be determined using the method of Section 10.5.7 of \cite{Nielsen}. By doing Gaussian elimination and permuting the columns/qubits, we can obtain the parity check matrix in standard or systematic form and then get the generating set. 
Essentially, with the method, we have $K$ logical qubits when we identify $K$ independent logical $X$ generators and $K$ independent logical $Z$ generators. The $K$ pairs of logical generators are such that they are only not orthogonal to each other. The dimension of the CSS code is thus $K = N - rank(\mathbf{H}_{\mathrm{X}}) - rank(\mathbf{H}_{\mathrm{Z}}).$
Arranged as $\mathbf{L}_{\mathrm{X}}$ and $\mathbf{L}_{\mathrm{Z}}$ matrices, they exactly overlap only at the $K$ positions ($\mathbf{L}_{\mathrm{Z}} \mathbf{L}_{\mathrm{X}}^\mathsf{T} = \mathbf{I}_K$) and also are orthogonal to the respective parity check matrix, i.e.,
$\mathbf{L}_{\mathrm{Z}} \mathbf{H}_{\mathrm{X}}^\mathsf{T} = \mathbf{0}$ and $\mathbf{L}_{\mathrm{X}} \mathbf{H}_{\mathrm{Z}}^\mathsf{T} = \mathbf{0}$.
A general recipe to obtain such pairs of PCMs is to use two classical dual-containing codes $\mathcal{C}_{\mathrm{X}}$ and $\mathcal{C}_{\mathrm{Z}}$ with the respective PCMs $\mathbf{H}_{\mathrm{X}}$ and $\mathbf{H}_{\mathrm{Z}}$, respectively, such that the dual ($\mathcal{C}^\perp$) of one code is contained in the other code, i.e., 
\begin{align}
    \mathcal{C}_{\mathrm{X}}^\perp \subseteq \mathcal{C}_{\mathrm{Z}},\\
    \mathcal{C}_{\mathrm{Z}}^\perp \subseteq \mathcal{C}_{\mathrm{X}}.
\end{align}
Here, $\mathcal{C}_{\mathrm{X}}^{\perp}$ ($\mathcal{C}_{\mathrm{Z}}^{\perp}$) denotes the dual code whose generator matrix is $\mathbf{H}_{\mathrm{X}}$ ($\mathbf{H}_{\mathrm{Z}}$).
The logical operators $\mathbf{\bar{X}}$ ($\mathbf{\bar{Z}}$) for the CSS codes are given by $\mathcal{L}_{\mathrm{X}}=\mathcal{C}_{\mathrm{X}}/\mathcal{C}_{\mathrm{Z}}^{\perp}$ ($\mathcal{L}_{\mathrm{Z}}=\mathcal{C}_{\mathrm{Z}}/\mathcal{C}_{\mathrm{X}}^{\perp}$). The dual-containing property on classical codes $\mathcal{C}_\text{X}$ and $\mathcal{C}_\text{Z}$, $\mathcal{C}_\text{X}^\perp \subseteq \mathcal{C}_\text{Z}$ implies that the generator matrices in the classical sense are in the form $\mathbf{G}_\mathrm{Z}=\left[\begin{smallmatrix} \mathbf{H}_{\mathrm{X}}\\ \mathbf{L}_{\mathrm{X}} \end{smallmatrix}\right]$. 
 The fact $\mathbf{G}_\mathrm{Z} \mathbf{H}_\mathrm{Z}^\mathsf{T}=\mathbf{0}$ implies $\mathbf{H}_\mathrm{X}\mathbf{H}_\mathrm{Z}^\mathsf{T}=\mathbf{0}$ and $\mathbf{L}_\mathrm{X}\mathbf{H}_\mathrm{Z}^\mathsf{T}=\mathbf{0}$. 
 First, error patterns $\in \text{rowspace} (\mathbf{H}_\mathrm{X})$ correspond to $\text{X}$-stabilizers and do not harm the decoder. In contrast, error patterns $\mathbf{\ell} \in \text{rowspace}(\mathbf{L}_\mathrm{X})$ are codewords. They correspond to $\text{X}$-logical errors (and thus will also be referred to as \emph{logical codewords}). 
 The minimum distance is determined by the lowest Hamming weight of respective logical codewords. Here, the weight of a Pauli operator is defined as the number of ones (support) for a non-identity codeword. The overall minimum distance of a CSS code is given by $d = \min(d_{\mathrm{X}},d_{\mathrm{Z}})$. If the minimum distance is strictly greater than the stabilizer generator weight, then the QLDPC code is referred to as a degenerate code, otherwise, it is a non-degenerate code.
More generally, if the combined PCM is of rank $N-K$ and both $\mathbf{L}_{\mathrm{X}}$ and $\mathbf{L}_{\mathrm{Z}}$ logical matrices are of rank $K$ each, there are $2^{N-K}(4^K-1) = \mathcal{O}(2^{N+K})$ logical codewords~\cite{webster2024CodeConstruction_Evolutionary}. Since the number of non-trivial logical operators grows exponentially in $N + K$; finding the minimum distance of stabilizer codes is computationally intractable and an NP-hard problem \cite{Vardy97,Pryadko_DistVerify_QLDPC,kapshikar_QminD_hardness_2023}. Hence, we acknowledge that completely characterizing the structure and distribution of the low-weight codewords is a daunting challenge.   
\section{QLDPC Code Construction}
\label{sec_QLDPC_Code}
The CSS construction provides a natural way to build quantum codes from good classical codes. One such method is to take the Kronecker product of classical PCMs to obtain hypergraph product (HGP), balanced product, or lifted product (LP) codes. In this paper, we will focus on the LP-QLDPC codes obtained from fully connected type-1 protograph classical LDPC codes.  
\subsection{Lifted Product Quantum LDPC Codes}
\label{subsec_LPCodes}
Consider the construction of quasi-cyclic lifted product quantum LDPC codes~\cite{panteleev2022quantumAlmostLinearMinD} using two classical quasi-cyclic base matrices $\mathbf{B}_1$ and $\mathbf{B}_2$ (as in Eq. \eqref{eq:B_and_H} in their CPM form) of size $m_1 \times n_1$  and  $m_2 \times n_2$, respectively: 
\begin{align}\begin{split}
    \mathbf{B}_{\mathrm{X}} & = \begin{bmatrix}
        \mathbf{B}_1 \;\;\otimes \mathbf{I}_{n_2} & \mathbf{I}_{m_1} \otimes \mathbf{B}_2^{\mathsf{*}}
    \end{bmatrix},\\
    \mathbf{B}_{\mathrm{Z}} & = \begin{bmatrix}
       \mathbf{I}_{n_1} \otimes \mathbf{B}_2 &  \; \; \mathbf{B}_1^{\mathsf{*}} \; \; \;\otimes  \mathbf{I}_{m_2}   \\
    \end{bmatrix}.
\label{eq:lifted_product_base_matrices}
\end{split}
\end{align}

The corresponding bit-flip and phase-flip PCMs $\mathbf{H_{\mathrm{X}}} = \mathbbm{L}(\mathbf{B}_{\mathrm{X}})$ and $\mathbf{H_{\mathrm{Z}}}  = \mathbbm{L}(\mathbf{B}_{\mathrm{Z}})$ are obtained by replacing the integer terms with the corresponding CPMs and -$\infty$ with zero matrices. 
They are indeed a generalization of hypergraph-product codes, which have a trivial lift operation with circulant size $L = 1$, and the base matrix $\mathbf{B}$ with $\alpha^{-\infty}=0$ and $\alpha^0=1$, thus giving the binary parity check matrix of the classical code. 

The symplectic inner product condition is also satisfied enforcing that $\mathbf{H_{\mathrm{X}}}$ and $\mathbf{H_{\mathrm{Z}}}$ parity check matrices are orthogonal, i.e., $\mathbf{H_{\mathrm{X}}}$ $\mathbf{H_{\mathrm{Z}}}^\textsc{T} = \mathbf{0}.$ 
Hence, the rowspace of $\mathbf{H_{\mathrm{X}}}$ essentially forms the stabilizer codewords with respect to $\mathbf{H_{\mathrm{Z}}}$ and vice-versa.
Using a single base matrix, i.e., $\mathbf{B}_1$ = $\mathbf{B}_2 = \mathbf{B}$ of size $m \times n$, we get the symmetric LP-QLDPC codes as in Eq. \eqref{eq:symmLP_base_matrices}. 

\begin{align}\begin{split}
    \mathbf{B}_{\mathrm{X}} & = \begin{bmatrix}
        \mathbf{B} \;\;\otimes \mathbf{I}_{n} & \mathbf{I}_{m} \otimes \mathbf{B}^{\mathsf{*}}
    \end{bmatrix},\\
    \mathbf{B}_{\mathrm{Z}} & = \begin{bmatrix}
       \mathbf{I}_{n} \otimes \mathbf{B} &  \; \; \mathbf{B}^{\mathsf{*}} \; \; \;\otimes  \mathbf{I}_{m}   \\
    \end{bmatrix}.
\label{eq:symmLP_base_matrices}
\end{split}
\end{align}
The corresponding parity check matrices $\mathbf{H_{\mathrm{X}}}$ and $\mathbf{H_{\mathrm{Z}}}$ each have $n_{\mathbf{H_\mathrm{X}}} = n_{\mathbf{H_\mathrm{Z}}} = L(n^2 + m^2)$ columns and $m_{\mathbf{H_\mathrm{X}}} = m_\mathbf{{H_\mathrm{Z}}} = L(nm)$ rows. 
Let us denote the classical code corresponding to the base matrix as $\mathcal{C}$ and the lifted product code obtained from the base matrix as $\mathcal{Q}$ with parameters $[L \times n,k,d_{\text{min}}^\mathcal{C}]$ and $[[N,K,d_{\text{min}}^\mathcal{Q}]]$, respectively. 

\subsection{Base Matrix of an LP-QLDPC Code}
Consider a base matrix (in exponent form) of a classical QC LDPC code with circulant size $L$, $m = 2$, and $n = 3$. 
\begin{equation}
\begin{split}
B =
\begin{bmatrix}
b_{1,1}&b_{1,2}&b_{1,3}\\
b_{2,1}&b_{2,2}&b_{2,3}
\end{bmatrix}
\end{split}
\label{eq:B_2_3}
\end{equation}
and 
\begin{equation}
\begin{split}
B^*=
\begin{bmatrix}
b_{1,1}^*&b_{2,1}^*\\
b_{1,2}^*&b_{2,2}^*\\
b_{1,3}^*&b_{2,3}^*
\end{bmatrix}
\end{split}
\label{eq:B*}
\end{equation}
where $b_{i,j} + b_{i,j}^* = 0, \mod L.$ A symmetric LP-QLDPC code using base matrix $B$ has the following representation in their simplified notation with CPM values: 
\begin{align}
B_{\mathrm{X}} & = 
\begin{bmatrix}
        B \otimes I_3 & I_2 \otimes B^{\mathsf{*}}  
\end{bmatrix},
\label{eq:BX}
\end{align} 
and 
\begin{align}
B_{\mathrm{Z}} & = \begin{bmatrix}
       {I}_{3} \otimes B & B^{\mathsf{*}} \otimes  {I}_2
    \end{bmatrix}.
    \label{eq:BZ}
\end{align} 
We have the left part of $B_{\mathrm{Z}}$ in Eq. \eqref{eq:BZ} as 
\begin{equation}
\arrayrulecolor{gray!50}
I_3 \otimes B =\left[
\begin{array}{;{2pt/2pt}c;{2pt/2pt}c;{2pt/2pt}c;{2pt/2pt}c;{2pt/2pt}c;{2pt/2pt}c;{2pt/2pt}c;{2pt/2pt}c;{2pt/2pt}c;{2pt/2pt}}
\hdashline[2pt/2pt]
b_{1,1} & b_{1,2} & b_{1,3} &
    &     &     &
    &     &     \\ \hdashline[2pt/2pt]
b_{2,1} & b_{2,2} & b_{2,3} &
    &     &     &
    &     &     \\ \hdashline[2pt/2pt]
    &     &     &
b_{1,1} & b_{1,2} & b_{1,3} &
    &     &     \\ \hdashline[2pt/2pt]
    &     &     &
b_{2,1} & b_{2,2} & b_{2,3} &
    &     &     \\ \hdashline[2pt/2pt]
    &     &     &
    &     &     &
b_{1,1} & b_{1,2} & b_{1,3} \\ \hdashline[2pt/2pt]
    &     &     &
    &     &     &
b_{2,1} & b_{2,2} & b_{2,3} \\ \hdashline[2pt/2pt]
  \end{array}
\right],
\label{Eq:BZ_left}
\end{equation}
where for simplicity, we denote the $\minf$ powers in the matrix array with empty cells. Observe that the left part of $B_{\mathrm{Z}}$ is structured as the $n$-fold concatenation of the base code. The right part of $B_{\mathrm{Z}}$ in Eq. \eqref{eq:BZ} is 
\begin{equation}
\arrayrulecolor{gray!50}
B^* \otimes I_2 =
\left[
\begin{array}{;{2pt/2pt}c;{2pt/2pt}c;{2pt/2pt}c;{2pt/2pt}c;{2pt/2pt}c;{2pt/2pt}}
\hdashline[2pt/2pt]
b_{1,1}^* &           & b_{2,1}^* &      \\ \hdashline[2pt/2pt]
          & b_{1,1}^* &           & b_{2,1}^* \\ \hdashline[2pt/2pt]
b_{1,2}^* &           & b_{2,2}^* &      \\ \hdashline[2pt/2pt]
          & b_{1,2}^* &           & b_{2,2}^* \\ \hdashline[2pt/2pt]
b_{1,3}^* &           & b_{2,3}^* &      \\ \hdashline[2pt/2pt]
          & b_{1,3}^* &           & b_{2,3}^* \\ \hdashline[2pt/2pt]
\end{array}
\right],
\label{Eq:BZ_right}
\end{equation}
where one can identify the transposed conjugate base matrix interleaved within the blocks. For this reason, we refer to the left part in Eq. \eqref{eq:BZ} as the `code' part of $B_{\mathrm{Z}}$ and the right part as the `transpose' part of $B_{\mathrm{Z}}$. 
Similarly, the left part of $B_{\mathrm{X}}$ in Eq. \eqref{eq:BX} also has the classical code base matrix interleaved within as
\begin{equation}
\arrayrulecolor{gray!50}
B \otimes I_3  =\left[
\begin{array}{;{2pt/2pt}c;{2pt/2pt}c;{2pt/2pt}c;{2pt/2pt}c;{2pt/2pt}c;{2pt/2pt}c;{2pt/2pt}c;{2pt/2pt}c;{2pt/2pt}c;{2pt/2pt}}
\hdashline[2pt/2pt]
b_{1,1} &         &         & b_{1,2} &         &        & b_{1,3} &         &         \\ \hdashline[2pt/2pt]
        & b_{1,1} &         &         & b_{1,2} &        &         & b_{1,3} &         \\ \hdashline[2pt/2pt]
        &         & b_{1,1} &         &         & b_{1,2}&         &         & b_{1,3} \\ \hdashline[2pt/2pt]
b_{2,1} &         &         & b_{2,2} &         &        & b_{2,3} &         &         \\ \hdashline[2pt/2pt]
        & b_{2,1} &         &         & b_{2,2} &        &         & b_{2,3} &         \\ \hdashline[2pt/2pt]
        &         & b_{2,1} &         &         & b_{2,2}&         &         & b_{2,3} \\ \hdashline[2pt/2pt]
\end{array}
\right],
\label{Eq:BX_left}
\end{equation}
and the right part of $B_{\mathrm{X}}$ in Eq. \eqref{eq:BX} is composed as 
\begin{equation}
\arrayrulecolor{gray!50}
I_2 \otimes B^* =
\left[
  \begin{array}{;{2pt/2pt}c;{2pt/2pt}c;{2pt/2pt}c;{2pt/2pt}c;{2pt/2pt}c;{2pt/2pt}}
\hdashline[2pt/2pt]
b_{1,1}^* & b_{2,1}^* &  & \\ \hdashline[2pt/2pt]
b_{1,2}^* & b_{2,2}^* &  &  \\ \hdashline[2pt/2pt]
b_{1,3}^* & b_{2,3}^* &     &     \\ \hdashline[2pt/2pt]
    &     & b_{1,1}^* & b_{2,1}^* \\ \hdashline[2pt/2pt]
    &     & b_{1,2}^* & b_{2,2}^* \\ \hdashline[2pt/2pt]
    &     & b_{1,3}^* & b_{2,3}^* \\ \hdashline[2pt/2pt]
  \end{array}
\right],
\label{Eq:BX_right}
\end{equation}
where one can notice the $m$-fold concatenation of the transposed conjugate base matrix. Together, we have the matrices corresponding $B_\mathrm{Z}$ and $B_\mathrm{X}$ as:
\begin{equation}
\arrayrulecolor{gray!50}
B_\mathrm{X}=\left[
\begin{array}{;{2pt/2pt}c;{2pt/2pt}c;{2pt/2pt}c;{2pt/2pt}c;{2pt/2pt}c;{2pt/2pt}c;{2pt/2pt}c;{2pt/2pt}c;{2pt/2pt}c;{2pt/2pt} ;{2pt/2pt}c;{2pt/2pt}c;{2pt/2pt}c;{2pt/2pt}c;{2pt/2pt}c;{2pt/2pt}}
\hdashline[2pt/2pt]
b_{1,1} &  &  & b_{1,2} &  &  & b_{1,3} &  &  & b_{1,1}^* & b_{2,1}^* &  &  \\ \hdashline[2pt/2pt]
  & b_{1,1} &  &  & b_{1,2} &  &  & b_{1,3} &  & b_{1,2}^* & b_{2,2}^* &  &  \\ \hdashline[2pt/2pt]
  &  & b_{1,1} &  &  & b_{1,2} &  &  & b_{1,3} & b_{1,3}^* & b_{2,3}^* &  &  \\ \hdashline[2pt/2pt]
b_{2,1} &  &  & b_{2,2} &  &  & b_{2,3} &  &  &  &  & b_{1,1}^* & b_{2,1}^* \\ \hdashline[2pt/2pt]
  & b_{2,1} &  &  & b_{2,2} &  &  & b_{2,3} &  &  &  & b_{1,2}^* & b_{2,2}^* \\ \hdashline[2pt/2pt]
  &  & b_{2,1} &  &  & b_{2,2} &  &  & b_{2,3} &  &  & b_{1,3}^* & b_{2,3}^* \\ \hdashline[2pt/2pt]
\end{array}
\right],
\label{Eq:BX_leftright}
\end{equation}

\begin{equation}
\arrayrulecolor{gray!50}
B_\mathrm{Z}=\left[
\begin{array}{;{2pt/2pt}c;{2pt/2pt}c;{2pt/2pt}c;{2pt/2pt}c;{2pt/2pt}c;{2pt/2pt}c;{2pt/2pt}c;{2pt/2pt}c;{2pt/2pt}c;{2pt/2pt} ;{2pt/2pt}c;{2pt/2pt}c;{2pt/2pt}c;{2pt/2pt}c;{2pt/2pt}c;{2pt/2pt}}
\hdashline[2pt/2pt]
b_{1,1} & b_{1,2} & b_{1,3} &  &  &  &  &  &  & b_{1,1}^* &  & b_{2,1}^* &   \\ \hdashline[2pt/2pt]
b_{2,1} & b_{2,2} & b_{2,3} &  &  &  &  &  &  &  & b_{1,1}^* &  & b_{2,1}^*  \\ \hdashline[2pt/2pt]
  &  &  & b_{1,1} & b_{1,2} & b_{1,3} &  &  &  & b_{1,2}^* &  & b_{2,2}^* &    \\ \hdashline[2pt/2pt]
  &  &  & b_{2,1} & b_{2,2} & b_{2,3} &  &  &  &  &  b_{1,2}^* &  & b_{2,2}^* \\ \hdashline[2pt/2pt]
  &  &  &  &  &  & b_{1,1} & b_{1,2} & b_{1,3} & b_{1,3}^* &  & b_{2,3}^* &    \\ \hdashline[2pt/2pt]  
  &  &  &  &  &  & b_{2,1} & b_{2,2} & b_{2,3} &  & b_{1,3}^* &  & b_{2,3}^*  \\ \hdashline[2pt/2pt]
\end{array}
\right].
\label{Eq:BZ_leftright}
\end{equation}

The girth $g$, defined as the shortest cycle length in its Tanner graph, is limited to $g=8$ for an LP-QLDPC code. There are length-8 cycles which are unavoidable as they are a result of the LP construction itself.
The constraint for cycles in terms of the base matrix given in \cite{arrayCode_fan,Fossorier04} is always satisfied in LP-QLDPC codes for the existence of length-8 cycles. For instance, in Eq. \eqref{Eq:BX_leftright}, 
$b_{1,1}  - b_{1,1}^*  + b_{1,2}^* - b_{1,1} + b_{2,1} - b_{1,2}^* + b_{1,1}^* - b_{2,1} \mod L = 0$. 
This can be easily generalized to identify unavoidable length-8 cycles that loop back around starting from the left/code part to the right/transpose side twice \cite{raveendran2021trapping}. 

\textit{Example 1 Continued: }
Using $B$ matrix in Eq: \eqref{Eq:Example1BaseMX} as the base matrix, the corresponding LP-QLDPC code is obtained as 
\begin{equation}
\arrayrulecolor{gray!50}
B_\mathrm{X}=\left[
\begin{array}{;{2pt/2pt}c;{2pt/2pt}c;{2pt/2pt}c;{2pt/2pt}c;{2pt/2pt}c;{2pt/2pt}c;{2pt/2pt}c;{2pt/2pt}c;{2pt/2pt}c;{2pt/2pt} ;{2pt/2pt}c;{2pt/2pt}c;{2pt/2pt}c;{2pt/2pt}c;{2pt/2pt}c;{2pt/2pt}}
\hdashline[2pt/2pt]
1 &  &  & 2 &  &  & 4 &  &  & 6 & 1 &  &  \\ \hdashline[2pt/2pt]
  & 1 &  &  & 2 &  &  & 4 &  & 5 & 2 &  &  \\ \hdashline[2pt/2pt]
  &  & 1 &  &  & 2 &  &  & 4 & 3 & 4 &  &  \\ \hdashline[2pt/2pt]
6 &  &  & 5 &  &  & 3 &  &  &  &  & 6 & 1 \\ \hdashline[2pt/2pt]
  & 6 &  &  & 5 &  &  & 3 &  &  &  & 5 & 2 \\ \hdashline[2pt/2pt]
  & & 6  &  &  & 5  &  &  & 3 &  &  & 3 & 4 \\ \hdashline[2pt/2pt]
\end{array}
\right],
\label{Eq:BX_leftrightExample}
\end{equation}

\begin{equation}
\arrayrulecolor{gray!50}
B_\mathrm{Z}=\left[
\begin{array}{;{2pt/2pt}c;{2pt/2pt}c;{2pt/2pt}c;{2pt/2pt}c;{2pt/2pt}c;{2pt/2pt}c;{2pt/2pt}c;{2pt/2pt}c;{2pt/2pt}c;{2pt/2pt} ;{2pt/2pt}c;{2pt/2pt}c;{2pt/2pt}c;{2pt/2pt}c;{2pt/2pt}c;{2pt/2pt}}
\hdashline[2pt/2pt]
1 & 2 & 4 &  &  &  &  &  &  & 6 &  & 1 &    \\ \hdashline[2pt/2pt]
6 & 5 & 3 &  &  &  &  &  &  &  & 6 &  & 1  \\ \hdashline[2pt/2pt]
  &  &  & 1 & 2 & 4  &  &  &  & 5 &  & 2 &    \\ \hdashline[2pt/2pt]
  &  &  & 6 & 5 & 3 &  &  &  &  &  5 &  & 2 \\ \hdashline[2pt/2pt]
  &  &  &  &  &  & 1 & 2 & 4 & 3 &  & 4 &    \\ \hdashline[2pt/2pt]  
  &  &  &  &  &  & 6 & 5 & 3 &  & 3 &  & 4  \\ \hdashline[2pt/2pt]
\end{array}
\right],
\label{Eq:BZ_leftrightExample}
\end{equation}

\subsection{Orthogonality Condition for Base Matrices}
\label{sec:OrthogonalityBasematrices}
A necessary and sufficient condition for a pair of QC-LDPC codes to satisfy the symplectic inner product condition is given in \cite{imai_qc_codes}. We revisit the even multiplicity condition described therein first.  
\begin{definition}[Even multiplicity]
A vector \(\mathbf{\mu} = (\mu_1, \mu_2, \ldots, \mu_n)\) is said to have even multiplicity if, for every integer \(k \in \mathbb{Z}\), the number of indices \(i \in \{1, 2, \ldots, n\}\) such that \(\mu_i = k\) is even.
\end{definition} 
By definition, the vector has even multiplicity if each integer entry in $\mathbf{\mu}$ appears an even number of times while discounting the $\minf$ entries or the empty cells if any.

\textit{Even multiplicity and orthogonal base matrices:}
Let $B$ be a base matrix in exponent form of size $m \times n $ with circulant size $L$, where the entries $ b_{i,j} \in \{\minf, 0, 1, \ldots, L-1\}$.
Define the $i$-th row vector $\mathbf{r}_i$ for $i = 1, 2, \ldots, m$ as $    \mathbf{r}_i = \begin{bmatrix} b_{i,1}, & b_{i,2}, & \cdots, & b_{i,n} \end{bmatrix}.$
If the modulo-$L$ difference of any two distinct rows $\mathbf{r}_i$ and $\mathbf{r}_j$: i.e., $\mathbf{r} = \mathbf{r}_i - \mathbf{r}_j$, modulo-$L$, has even multiplicity, then the corresponding binary matrices after lifting $\mathbbm{L}(\mathbf{r}_i)$ and $\mathbbm{L}(\mathbf{r}_j)$ are orthogonal\cite{imai_qc_codes,galindo_quasi-cyclic_2018}. 
The minus $-$ operation is performed modulo-$L$ and for the $\infty$ entries is as follows: $b_{i,j} - \infty = \infty - b_{i,j} = \infty - \infty := -\infty.$ 

\begin{te}
    LP-QLDPC code PCMs $\mathbf{H_{\mathrm{X}}}$ and $\mathbf{H_{\mathrm{Z}}}$ are orthogonal if the rows of quasi-cyclic base matrices $B_{\mathrm{X}}$ and $B_{\mathrm{Z}}$ have even multiplicity.
    \label{Th:EvenMu}
\end{te}

\textit{Proof:} 
To prove the orthogonality property for quasi-cyclic LP-QLDPC code parity check matrices, one can check the even multiplicity property for the base matrices $B_{\mathrm{X}}$ and $B_{\mathrm{Z}}$. Their rows obey the condition for orthogonality - the even multiplicity of integers in the modulo-$L$ differences of the rows, where $L$ is the circulant size. Take any pair of rows such that $\mathbf{r}_i$ is taken from $B_{\mathrm{X}}$, and $\mathbf{r}_j$ from $B_{\mathrm{Z}}$. 
These rows can be split based on the left part and right part - with exactly one overlap in each of the parts. 
Thus, the difference $\mathbf{r} = \mathbf{r}_i - \mathbf{r}_j$, modulo-$L$ will only have two integer terms. 
Suppose that the two overlapping terms are equal; then, the resulting difference vector has a pair of zeros.
If the overlapping terms $p$ and $q$ are not equal, $p \neq q$, then the corresponding terms in the difference appear as $p - q$ and $q^* - p^*$, respectively. Since the conjugate terms are $q^* = L-q$ and $p^* = L-p$, the overlapping terms are indeed equal, satisfying the even multiplicity condition. Therefore,  the two parity check matrices are orthogonal to each other. \qed

\section{Minimum Distance of LP-QLDPC Codes}
\label{sec:MinDistLPCodes}
In the following, we examine the symmetric LP-QLDPC codes constructed from type-1 quasi-cyclic base matrices to understand how certain combinations of CPM values reduce the quantum minimum distance $d_{\text{min}}^\mathcal{Q}$ from the minimum distance of the base code $d_{\text{min}}^\mathcal{C}$. We first show how the minimum distance is limited to the Hamming weight of the stabilizer generators. We also prove that using a base matrix with only two rows ($m = 2$ as in Eq. \eqref{eq:B_2_3}) always limits the LP-QLDPC code to have a minimum distance $d_{\text{min}}^\mathcal{Q} \le n+2$. In general, we are interested in the following question: What constraints on the base code of an LP-QLDPC code reduce its minimum distance? This will guide us toward code construction recipes guarantee degeneracy for the obtained LP-QLDPC code. First, we provide examples of the choices of base codes of different sizes: $m$ and $n$ to demonstrate the reduction of the minimum distance of LP-QLDPC codes.\\
\textit{Example 1 Continued:} Consider the LP-QLDPC code constructed from the example base matrix we saw earlier with $m = 2, n = 3, L =7$,
\begin{equation}
B=\begin{bmatrix}
1&2&4\\
6&5&3
\end{bmatrix}.
\label{eq: B_2_3}
\end{equation}
We start with a classical LDPC code $\mathcal{C}: [21,8,6]$. Hence, the desired minimum distance of LP code is $d_{\text{min}}^{\mathcal{C}} = 6$. However, we verified that the $d_{\text{min}}^{\mathcal{Q}} = m+n = 5$. 
Note that for quasi-cyclic base matrices, it is enough to look at equivalent base matrices expressed in canonical form. Two base matrices are equivalent if their respective Tanner graphs are identical upon variable and/or check node permutations. Given a base matrix, $B$ exchanging rows/columns, or adding a fixed integer to each element in row/column maintains equivalence. The rows and columns of the quasi-cyclic code can be reordered to put 0s in the first row and column to make an equivalent parity check matrix making the analysis easier. For instance, we have the equivalent matrix \begin{equation}
B=\begin{bmatrix}
0&0&0\\
0&1&3
\end{bmatrix}.
\end{equation}  
\begin{ex}
Classical code $\mathcal{C}: [104,30,14]$ with the base matrix given in Example 3 of \cite{Tanner_QC}. 
We obtain the base matrix for the classical LDPC code $\mathcal{C}: [104,30,14]$. The base matrix $B$ in Eq. \eqref{eq: B_3_4_00} corresponds to a $(3,4)$ regular LDPC code with circulant size $L = 26$. The minimum distance of the base code is $d_{\text{min}}^{\mathcal{C}} = 14$. 
\begin{equation}
B=\begin{bmatrix}
0&0&0&0\\
0&6&4&10\\
0&8&14&22
\end{bmatrix}.
\label{eq: B_3_4_00}
\end{equation}
We now construct a symmetric LP-QLDPC code from Eq. \eqref{eq: B_3_4_00} to get a [[650, 50, 7]] code. Here, we have $N = 26 \times (3^2 + 4^2)$ physical qubits encoding $K = 50$ logical qubits with a minimum distance $d_{\text{min}}^{\mathcal{Q}} = 7$ instead of 14. Note that the quantum minimum distance is again equal to $m + n$.
\end{ex}

\begin{ex} 
Classical code $\mathcal{C}:[65,16,16]$ with the base matrix given in Eq. \eqref{eq: B_4_5} having $m = 4$ rows and $n = 5$ columns with circulant size $L = 13$. Here, the minimum distance is $d_{\text{min}}^{\mathcal{C}} = 16$. 
 
\begin{equation}
\begin{split}
B = \begin{bmatrix}
0&0&0&0&0\\
0&1&11&8&9\\
0&4&5&6&10\\
0&10&6&2&12 	
\end{bmatrix}.
\end{split}
\label{eq: B_4_5}
\end{equation}

We construct a symmetric LP-QLDPC code from $B$ in Eq. \eqref{eq: B_4_5} to make a [[533, 37, 9]] quantum code. Here, we have $N = 13 \times (4^2 + 5^2)$ physical qubits encoding $K = 37$ logical qubits with a minimum distance $d_{\text{min}}^{\mathcal{Q}} = 9$ instead of the minimum distance $d_{\text{min}}^{\mathcal{C}} = 16$.    
\end{ex}

As shown in these examples, after performing the Kronecker product and lifting, the LP-QLDPC codes all have $d_{\text{min}}^\mathcal{Q}  < d_{\text{min}}^\mathcal{C}$. Notably, the reduction in quantum minimum distance is such that $d_{\text{min}}^\mathcal{Q}$ is, in fact, equal to the respective stabilizer weights, resulting in non-degenerate LP-QLDPC codes. 
The following subsections summarize the necessary conditions to avoid having such critical stabilizer-weight logical codewords. 

\subsection{Stabilizer Codewords} 
For quantum CSS codes, the commutativity of the stabilizers implies that the matrices $\mathbf{H_{\mathrm{X}}}$ and $\mathbf{H_{\mathrm{Z}}}$ are orthogonal. 
Alternatively, we have also seen in subsection \ref{sec:OrthogonalityBasematrices}, how the base matrices of $\mathbf{H_{\mathrm{X}}}$ and $\mathbf{H_{\mathrm{Z}}}$ satisfy the even multiplicity constraint. This means that the row space of $\mathbf{H_{\mathrm{X}}}$ essentially forms the stabilizer codewords with respect to $\mathbf{H_{\mathrm{Z}}}$, and vice versa. However, these stabilizer codewords do not contribute to the minimum distance of the quantum code. The stabilizer codewords of the LP-QLDPC codes with the lowest weight are in fact the rows of the parity check matrices, having Hamming weight equal to $m+n$. 
Utilizing the quasi-cyclic property, we observe that any modulo-$L$ shift of a stabilizer remains a valid stabilizer. Therefore, the block representation provides a concise and structured way to express quasi-cyclic stabilizer codewords. We can express the stabilizer codewords in the block form with support in both the `code' and `transposed' parts of the parity check matrix as follows:
\begin{equation}
    \bm{s} =\left( 
    \begin{array}{c|c}
    \bm{s}_{\mathrm{C}} &\bm{s}_{\mathrm{T}} 
    \end{array}
     \right).
     \label{eq:stabilizers}
\end{equation} 
An example of stabilizer codeword in its block form with respect to $B_{\mathrm{X}}$ is the first row of $B_{\mathrm{Z}}$ as in Eq. \eqref{Eq:BZ_leftright} represented as: 
\begin{equation}
\bm{s}=\left( \arrayrulecolor{gray!50}
\left[
\begin{array}{;{2pt/2pt}c ;{2pt/2pt}c ;{2pt/2pt}c ;{2pt/2pt}c ;{25pt/2pt}c
;{2pt/2pt}c ;{2pt/2pt}c ;{2pt/2pt}c ;{25pt/2pt}c 
;{2pt/2pt}c ;{2pt/2pt}c ;{2pt/2pt}c  ;{25pt/2pt}c  
;{2pt/2pt}c  ;{2pt/2pt}c  ;{2pt/2pt}c  ;{2pt/2pt}}
\hdashline[2pt/2pt]
 b_{1,1} & b_{1,2} & \ldots & b_{1,n} &              &         &        &         &  
         &         &        &         & 
         &         &        &         \\
\hdashline[2pt/2pt]
\end{array}
\right] |
\left[
 \begin{array}{;{2pt/2pt}c ;{2pt/2pt}c ;{2pt/2pt}c  ;{25pt/2pt}c
 ;{2pt/2pt}c ;{2pt/2pt}c ;{25pt/2pt}c  
 ;{2pt/2pt}c ;{2pt/2pt}c ;{2pt/2pt}}
\hdashline[2pt/2pt]
b_{1,1}^* &  & \ldots & 
b_{2,1}^* &  & \ldots & 
b_{m,1}^* &  & \ldots \\
\hdashline[2pt/2pt]
\end{array}
\right]
\right).
\end{equation}
Similarly, with respect to $B_{\mathrm{Z}}$, we have a stabilizer codeword as the first row of $B_{\mathrm{X}}$ as in Eq. \eqref{Eq:BX_leftright}:
\begin{equation}
\bm{s}=\left( \arrayrulecolor{gray!50}
\left[
\begin{array}{;{2pt/2pt}c ;{2pt/2pt}c ;{2pt/2pt}c ;{2pt/2pt}c ;{25pt/2pt}c
;{2pt/2pt}c ;{2pt/2pt}c ;{2pt/2pt}c ;{25pt/2pt}c 
;{2pt/2pt}c ;{2pt/2pt}c ;{2pt/2pt}c  ;{25pt/2pt}c  
;{2pt/2pt}c  ;{2pt/2pt}c  ;{2pt/2pt}c  ;{2pt/2pt}}
\hdashline[2pt/2pt]
 b_{1,1} &         &        &         &      b_{1,2} &         &        &         &  
 \ldots  &         &        &         & 
 b_{1,n} &         &        &         \\
\hdashline[2pt/2pt]
\end{array}
\right] |
\left[
 \begin{array}{;{2pt/2pt}c ;{2pt/2pt}c ;{2pt/2pt}c ;{2pt/2pt}c  ;{25pt/2pt}c
 ;{2pt/2pt}c ;{2pt/2pt}c ;{25pt/2pt}c  
 ;{2pt/2pt}c ;{2pt/2pt}c ;{2pt/2pt}}
\hdashline[2pt/2pt]
b_{1,1}^* & b_{2,1}^* & \ldots& b_{m,1}^* & 
          &  &  & 
          &  &  \\
\hdashline[2pt/2pt]
\end{array}
\right]
\right).
\end{equation}
In the block representation, the code part has $n^2$ cells divided into $n$ blocks with $n$ cells each. On the transpose part, there are $m^2$ cells that are divided into $m$ blocks with $m$ cells each. With respect to the binary PCM of the LP-QLDPC code, each cell is of length $L$ bits. 

We further represent the codewords in their \emph{canonical} block representation to book-keep and enumerate them using the quasi-cyclic property. We circularly shift the support of the codeword such that the first shift value is \emph{zero} in the block representation. Such a representation makes the enumeration/location of quasi-cyclic codewords and their cyclic shifts consistent. 

\subsection{Logical Codewords} 
We first obtain the logical $X$ and $Z$ codeword generators by performing Gaussian elimination on $\mathbf{H_{\mathrm{X}}}$ and $\mathbf{H_{\mathrm{Z}}}$ \cite{Gottesman97}. 
However, the minimum distance for a stabilizer code is defined as the smallest Hamming weight of codewords in the rowspace of the logical generators, and that is not in the rowspace of the PCMs. To tackle them systematically, we first enumerate the lowest-weight logical codewords from the nullspace of the classical parity check matrix using the impulse method \cite{Berrou02_1, Declercq2008}. 
In general, computing the minimum distance of a given linear code is an intractable non-polynomial problem. 
However, the algorithmic solution proposed in \cite{Berrou02_1, Declercq2008} allows, for short to moderate length LDPC codes, to obtain a tight estimation of the minimum distance. 
In this paper, we will assume that the minimum distances obtained using the impulse method are accurate.

Note that the LP construction preserves such \emph{low-weight} codewords, as they can also be identified in the LP-QLDPC code. A characteristic behavior is that their support (indices where the codeword vector is non-zero) is restricted in the `code' part. We can easily enumerate them by performing a Kronecker product of the minimal weight codewords obtained from the base matrix with $I_{n}$ and they will only have support in the code part as follows: 
\begin{equation}
    \bm{\ell} =\left( 
    \begin{array}{c|c}
    \bm{\ell}_{\mathrm{C}} &\bf{0}_{\mathrm{T}} 
    \end{array}
    \right),
    \label{eq:cw_classical}
\end{equation} 
where the indices $\mathrm{C}$  and $\mathrm{T}$ indicate respectively the code and transpose part of the LP-QLDPC matrix. 
Note that these logical codewords do not have support in the `transpose' part as indicated by $\bf{0}_{\mathrm{T}}$.

\subsection{Characterization of Lowest-weight Logical Codewords}
Enumerating the low-weight codewords of the LP-QLDPC code directly using the impulse method, we identify new codewords of lower weight; especially interesting are those with Hamming weight equal to the stabilizer weight. 
One common characteristic behavior among these codewords that appear in addition to the classical codewords is that they span the entire code length similar to the stabilizer codewords, taking the form 
\begin{equation}
\bm{\ell} = \left(
\begin{array}{c|c}
\bm{\ell}_{\mathrm{C}} & \bm{\ell}_{\mathrm{T}}
\end{array}
\right).
\label{eq:cw_newlogical}
\end{equation}
However, they are not in the rowspace of the PCMs, which means their canonical block representation is different from any of the stabilizers. Having a non-zero transpose part also indicates that they are, unlike the codewords in Eq. \eqref{eq:cw_classical}, not originating from the classical LDPC code. 

To generalize and identify these new codewords, we check the orthogonality condition of $\bm{\ell}$ in Eq. \eqref{eq:cw_newlogical} with the corresponding $B_X$ and $B_Z$ matrices of type-1 LP-QLDPC code to see how the even multiplicity condition is satisfied. To facilitate this analysis, we impose the following assumption on the structure of such a logical codeword along with the reasoning:\\
\emph{Assumption}: The logical codewords with the Hamming weight equal to the stabilizer weight span the entire code length with equal support on all the blocks. 
Both the lowest weight stabilizers as well as the lowest weight codewords arising from the classical base matrix have unequal support on the blocks. Moreover, the lowest-weight classical codewords do not span the entire code length. Based on this restriction, we can only place one integer exponent in the cells in each of the $n$ blocks in the code part, as well as in each of the $m$ blocks in the transpose part - making the codeword weight equal to $m+n$. 
We can reorder the exact position of the integer exponent in each block by permutation of columns and rows of the base matrix. Hence, the two general canonical block forms of such logical codewords we discuss are:  

\begin{equation}
\bm{\ell}=\arrayrulecolor{gray!50}
\left[
 \begin{array}{;{2pt/2pt}c ;{2pt/2pt}c ;{2pt/2pt}c ;{2pt/2pt}c ;{25pt/2pt}c
 ;{2pt/2pt}c ;{2pt/2pt}c ;{2pt/2pt}c ;{25pt/2pt}c 
 ;{2pt/2pt}c ;{2pt/2pt}c ;{2pt/2pt}c  ;{25pt/2pt}c  
 ;{2pt/2pt}c  ;{2pt/2pt}c  ;{2pt/2pt}c  ;{2pt/2pt} 
 ;{2pt/2pt}c ;{2pt/2pt}c ;{2pt/2pt}c  ;{25pt/2pt}c
 ;{2pt/2pt}c ;{2pt/2pt}c ;{25pt/2pt}c  
 ;{2pt/2pt}c ;{2pt/2pt}c ;{2pt/2pt}}
\hdashline[2pt/2pt]
 l_1  &  &  &  & &  l_2  & & & & & \ldots   & & & & & l_n  &   &  & l_{n+1}  &  & l_{n+2} &   & l_{n+m} & & \\
\hdashline[2pt/2pt]
\end{array}
\right]
\end{equation}
and 
\begin{equation}
\bm{\ell}=\arrayrulecolor{gray!50}
\left[
 \begin{array}{;{2pt/2pt}c ;{2pt/2pt}c ;{2pt/2pt}c ;{2pt/2pt}c ;{25pt/2pt}c
 ;{2pt/2pt}c ;{2pt/2pt}c ;{2pt/2pt}c ;{25pt/2pt}c 
 ;{2pt/2pt}c ;{2pt/2pt}c ;{2pt/2pt}c  ;{25pt/2pt}c  
 ;{2pt/2pt}c  ;{2pt/2pt}c  ;{2pt/2pt}c  ;{2pt/2pt} 
 ;{2pt/2pt}c ;{2pt/2pt}c ;{2pt/2pt}c  ;{25pt/2pt}c
 ;{2pt/2pt}c ;{2pt/2pt}c ;{25pt/2pt}c  
 ;{2pt/2pt}c ;{2pt/2pt}c ;{2pt/2pt}}
\hdashline[2pt/2pt]
   &  &  & l_1 &  &   & & l_2 & & \ldots &   & & l_n & & &  & l_{n+1}  &  &  &  & l_{n+2} &   &  & & l_{n+m} \\
\hdashline[2pt/2pt]
\end{array}
\right].
\end{equation}
Given the block exponent representation for a base matrix 
\begin{equation}
B = \begin{bmatrix}
b_{1,1} & b_{1,2} & \cdots & b_{1,n} \\
b_{2,1} & b_{2,2} & \cdots & b_{2,n} \\
\vdots & \vdots & \ddots & \vdots \\
b_{m,1} & b_{m,2} & \cdots & b_{m,n}
\end{bmatrix},
\end{equation}
let us define the \(i\)-th row vector \(\mathbf{r}_i\) and the \(j\)-th column vector \(\mathbf{c}_j\) as follows: 
The \(i\)-th row vector \(\mathbf{r}_i\) for \(i = 1, 2, \ldots, m\): 
\begin{equation}
    \label{eq:row_ri}
    \mathbf{r}_i = \begin{pmatrix} b_{i,1}, & b_{i,2}, & \cdots &, b_{i,n} \end{pmatrix}.
\end{equation}
Similarly, the $j$-th column vector $\mathbf{c}_j$ for $j = 1, 2, \ldots, n$ is 
\begin{equation}
    \label{eq:col_cj}
    \mathbf{c}_j = \begin{pmatrix} b_{1,j}, & b_{2,j}, & \ldots &, b_{m,j} \end{pmatrix}^\mathsf{T}.
\end{equation}
In addition, in terms of the conjugate entries, we define   
\begin{equation}
\label{eq:row_rjConj}
\mathbf{r}_i^* = \begin{pmatrix} b_{i,1}^*, & b_{i,2}^*, & \ldots &, b_{i,n}^* \end{pmatrix} 
\end{equation}
and 
\begin{equation}
\label{eq:col_cjConj}
\mathbf{c}_j^* = \begin{pmatrix} b_{1,j}^* & b_{2,j}^* & \ldots & b_{m,j}^* \end{pmatrix}^\mathsf{T}.
\end{equation}
For simplicity and without loss of generality to obtain a working example, let us take a base matrix of dimension $m \times n  = 3 \times 4$ as 
$$
B = \begin{bmatrix}
b_{1,1} & b_{1,2} & b_{1,3} & b_{1,4} \\
b_{2,1} & b_{2,2} & b_{2,3}& b_{2,4} \\
b_{3,1} & b_{3,2} & b_{3,3}& b_{3,4}
\end{bmatrix}.$$
Based on our assumption, the logical codewords with the minimum Hamming weight that span the entire code with equal support of all the blocks have the block form  
\begin{equation}
\bm{\ell} =
\left(
  \begin{array}{c|c}
\bm{\ell}_{\mathrm{C}} & \bm{\ell}_{\mathrm{T}}
     \end{array}
\right)
\label{eq:cw_newlogical_example}
\end{equation}
where 
\begin{equation}
\bm{\ell}_\mathrm{C}=\arrayrulecolor{gray!50}
\left[
 \begin{array}{;{2pt/2pt}c ;{2pt/2pt}c ;{2pt/2pt}c ;{2pt/2pt}c ;{25pt/2pt}c
 ;{2pt/2pt}c ;{2pt/2pt}c ;{2pt/2pt}c ;{25pt/2pt}c 
 ;{2pt/2pt}c ;{2pt/2pt}c ;{2pt/2pt}c  ;{25pt/2pt}c  
 ;{2pt/2pt}c  ;{2pt/2pt}c  ;{2pt/2pt}c  ;{2pt/2pt}}
\hdashline[2pt/2pt]
 a  &  &  &  & &  b  & & & & & c   & & & & & d   \\
\hdashline[2pt/2pt]
\end{array}
\right]
\end{equation}
and 
\begin{equation}
\bm{\ell}_\mathrm{T}=\arrayrulecolor{gray!50}
\left[
 \begin{array}{;{2pt/2pt}c ;{2pt/2pt}c ;{2pt/2pt}c  ;{25pt/2pt}c
 ;{2pt/2pt}c ;{2pt/2pt}c ;{25pt/2pt}c  
 ;{2pt/2pt}c ;{2pt/2pt}c ;{2pt/2pt}}
\hdashline[2pt/2pt]
  &  &  x &  & y &   & z & & \\
\hdashline[2pt/2pt]
\end{array}
\right].
\end{equation}
Based on this representation, let us denote the vector corresponding to the elements $a, b,c,$ and $d$ as a row vector of length $4$ (in general $n$): 
\begin{equation}
    \mathbf{u} = [a, b, c, d].
\end{equation}
Also, let vectors $\mathbf{x},\mathbf{y},$ and  $\mathbf{z}$ represent row vectors $[x,x,x,x]$, $[y,y,y,y]$, and  $[z,z,z,z]$, respectively.\\
Let us check the orthogonality constraint by taking the minus operation (modulo-$L$) on the `code' part 
\begin{equation}
\bm{\ell}_\mathrm{C}=\arrayrulecolor{gray!50}
\left[
 \begin{array}{;{2pt/2pt}c ;{2pt/2pt}c ;{2pt/2pt}c ;{2pt/2pt}c ;{25pt/2pt}c
 ;{2pt/2pt}c ;{2pt/2pt}c ;{2pt/2pt}c ;{25pt/2pt}c 
 ;{2pt/2pt}c ;{2pt/2pt}c ;{2pt/2pt}c  ;{25pt/2pt}c  
 ;{2pt/2pt}c  ;{2pt/2pt}c  ;{2pt/2pt}c  ;{2pt/2pt}}
\hdashline[2pt/2pt]
 a  &  &  &  & &  b  & & & & & c   & & & & & d   \\
\hdashline[2pt/2pt]
\end{array}
\right]
\end{equation}
 with the rows on the code part of the LP-QLDPC code of $B_{\mathrm{X}}$ matrix:= $B \otimes {I}_4$. The first row of $B \otimes {I}_4$ is $$\arrayrulecolor{gray!50}
\left[
 \begin{array}{;{2pt/2pt}c ;{2pt/2pt}c ;{2pt/2pt}c ;{2pt/2pt}c ;{25pt/2pt}c
 ;{2pt/2pt}c ;{2pt/2pt}c ;{2pt/2pt}c ;{25pt/2pt}c 
 ;{2pt/2pt}c ;{2pt/2pt}c ;{2pt/2pt}c  ;{25pt/2pt}c  
 ;{2pt/2pt}c  ;{2pt/2pt}c  ;{2pt/2pt}c  ;{2pt/2pt}}
\hdashline[2pt/2pt]
b_{1,1} &  &  &  & b_{1,2} & & & & b_{1,3} & &  & &b_{1,4} & & &  \\
\hdashline[2pt/2pt]
\end{array}
\right].$$ 
So, we have the two row vectors as the following:
\begin{equation}
\arrayrulecolor{gray!50}
\left[
 \begin{array}{;{2pt/2pt}c ;{2pt/2pt}c ;{2pt/2pt}c ;{2pt/2pt}c ;{25pt/2pt}c
 ;{2pt/2pt}c ;{2pt/2pt}c ;{2pt/2pt}c ;{25pt/2pt}c 
 ;{2pt/2pt}c ;{2pt/2pt}c ;{2pt/2pt}c  ;{25pt/2pt}c  
 ;{2pt/2pt}c  ;{2pt/2pt}c  ;{2pt/2pt}c  ;{2pt/2pt}}
\hdashline[2pt/2pt]
 a  &  &  &  & &  b  & & & & & c   & & & & & d   \\
\hdashline[2pt/2pt]
b_{1,1} &  &  &  & b_{1,2} & & & & b_{1,3} & &  & &b_{1,4} & & &  \\
\hdashline[2pt/2pt]
\hdashline[2pt/2pt]
\end{array}
\right]
\end{equation}
Observe that only one cell overlaps and will result in value $b_{1,1} - a$. The subsequent three rows also give a single term each as they are of the form:  
\begin{equation}
\arrayrulecolor{gray!50}
\left[
 \begin{array}{;{2pt/2pt}c ;{2pt/2pt}c ;{2pt/2pt}c ;{2pt/2pt}c ;{25pt/2pt}c
 ;{2pt/2pt}c ;{2pt/2pt}c ;{2pt/2pt}c ;{25pt/2pt}c 
 ;{2pt/2pt}c ;{2pt/2pt}c ;{2pt/2pt}c  ;{25pt/2pt}c  
 ;{2pt/2pt}c  ;{2pt/2pt}c  ;{2pt/2pt}c  ;{2pt/2pt}}
\hdashline[2pt/2pt]
 a  &  &  &  & &  b  & & & & & c   & & & & & d   \\
\hdashline[2pt/2pt]
 & b_{1,1} &  &  &  &  b_{1,2}& & &  & b_{1,3} &  & & & b_{1,4} & &  \\
\hdashline[2pt/2pt]
\hdashline[2pt/2pt]
\end{array}
\right],
\end{equation}

\begin{equation}
\arrayrulecolor{gray!50}
\left[
 \begin{array}{;{2pt/2pt}c ;{2pt/2pt}c ;{2pt/2pt}c ;{2pt/2pt}c ;{25pt/2pt}c
 ;{2pt/2pt}c ;{2pt/2pt}c ;{2pt/2pt}c ;{25pt/2pt}c 
 ;{2pt/2pt}c ;{2pt/2pt}c ;{2pt/2pt}c  ;{25pt/2pt}c  
 ;{2pt/2pt}c  ;{2pt/2pt}c  ;{2pt/2pt}c  ;{2pt/2pt}}
\hdashline[2pt/2pt]
 a  &  &  &  & &  b  & & & & & c   & & & & & d   \\
\hdashline[2pt/2pt]
 &  & b_{1,1} &  &  &  & b_{1,2} & &  &  & b_{1,3}  & & &  & b_{1,4} &  \\
\hdashline[2pt/2pt]
\hdashline[2pt/2pt]
\end{array}
\right], 
\end{equation}
and
\begin{equation}
\arrayrulecolor{gray!50}
\left[
 \begin{array}{;{2pt/2pt}c ;{2pt/2pt}c ;{2pt/2pt}c ;{2pt/2pt}c ;{25pt/2pt}c
 ;{2pt/2pt}c ;{2pt/2pt}c ;{2pt/2pt}c ;{25pt/2pt}c 
 ;{2pt/2pt}c ;{2pt/2pt}c ;{2pt/2pt}c  ;{25pt/2pt}c  
 ;{2pt/2pt}c  ;{2pt/2pt}c  ;{2pt/2pt}c  ;{2pt/2pt}}
\hdashline[2pt/2pt]
 a  &  &  &  & &  b  & & & & & c   & & & & & d   \\
\hdashline[2pt/2pt]
 &  &  & b_{1,1} &  &  & & b_{1,2}  &  &  &   & b_{1,3} & &  &  & b_{1,4} \\
\hdashline[2pt/2pt]
\hdashline[2pt/2pt]
\end{array}
\right].
\end{equation}
Overall, the minus operation results in the following four (in general $n$) terms: 
$b_{1,1} - a$, $b_{1,2} - b$, $b_{1,3} - c$, and $b_{1,4} - d$. Utilizing the row vector representation, we can write as $\mathbf{r}_1 - \mathbf{u}$, where $\mathbf{u} = [a,b,c,d].$ Overall, we have three (in general $m$) expressions: 
\begin{align}
\begin{split}
    \label{eq:rowleft1}
        & \mathbf{r}_1 - \mathbf{u} \\
        & \mathbf{r}_2 - \mathbf{u} \\ 
        & \mathbf{r}_3 - \mathbf{u}.
\end{split}
\end{align}
Recall on the transpose part, we have 
\begin{equation}
\bm{\ell}_\mathrm{T}=\arrayrulecolor{gray!50}
\left[
 \begin{array}{;{2pt/2pt}c ;{2pt/2pt}c ;{2pt/2pt}c  ;{25pt/2pt}c
 ;{2pt/2pt}c ;{2pt/2pt}c ;{25pt/2pt}c  
 ;{2pt/2pt}c ;{2pt/2pt}c ;{2pt/2pt}}
\hdashline[2pt/2pt]
  &  &  x &  & y &   & z & & \\
\hdashline[2pt/2pt]
\end{array}
\right].
\end{equation}
Applying minus operation with the first row of $\mathbf{I}_3 \otimes B^*$, that is $
\arrayrulecolor{gray!50}
\left[
 \begin{array}{;{2pt/2pt}c ;{2pt/2pt}c ;{2pt/2pt}c  ;{25pt/2pt}c
 ;{2pt/2pt}c ;{2pt/2pt}c ;{25pt/2pt}c  
 ;{2pt/2pt}c ;{2pt/2pt}c ;{2pt/2pt}}
\hdashline[2pt/2pt]
 b_{1,1}^* &  b_{2,1}^* &  b_{3,1}^* &  &  &   & & & \\
\hdashline[2pt/2pt]
\end{array}
\right]$, we have 

\begin{equation}
\arrayrulecolor{gray!50}
\left[
 \begin{array}{;{2pt/2pt}c ;{2pt/2pt}c ;{2pt/2pt}c  ;{25pt/2pt}c
 ;{2pt/2pt}c ;{2pt/2pt}c ;{25pt/2pt}c  
 ;{2pt/2pt}c ;{2pt/2pt}c ;{2pt/2pt}}
\hdashline[2pt/2pt]
  &  &  x &  & y &   & z & & \\
  \hdashline[2pt/2pt]
 b_{1,1}^* &  b_{2,1}^* &  b_{3,1}^* &  &  &   & & & \\
\hdashline[2pt/2pt]
\end{array}
\right]
\end{equation}
and result in a single term: $b_{3,1}^* - x.$ The subsequent three rows are of the form: 
\begin{equation}
\arrayrulecolor{gray!50}
\left[
 \begin{array}{;{2pt/2pt}c ;{2pt/2pt}c ;{2pt/2pt}c  ;{25pt/2pt}c
 ;{2pt/2pt}c ;{2pt/2pt}c ;{25pt/2pt}c  
 ;{2pt/2pt}c ;{2pt/2pt}c ;{2pt/2pt}}
\hdashline[2pt/2pt]
  &  &  x &  & y &   & z & & \\
  \hdashline[2pt/2pt]
 b_{1,2}^* &  b_{2,2}^* &  b_{3,2}^* &  &  &   & & & \\
\hdashline[2pt/2pt]
\end{array}
\right]
\end{equation}

\begin{equation}
\arrayrulecolor{gray!50}
\left[
 \begin{array}{;{2pt/2pt}c ;{2pt/2pt}c ;{2pt/2pt}c  ;{25pt/2pt}c
 ;{2pt/2pt}c ;{2pt/2pt}c ;{25pt/2pt}c  
 ;{2pt/2pt}c ;{2pt/2pt}c ;{2pt/2pt}}
\hdashline[2pt/2pt]
  &  &  x &  & y &   & z & & \\
  \hdashline[2pt/2pt]
 b_{1,3}^* &  b_{2,3}^* &  b_{3,3}^* &  &  &   & & & \\
\hdashline[2pt/2pt]
\end{array}
\right]
\end{equation}

\begin{equation}
\arrayrulecolor{gray!50}
\left[
 \begin{array}{;{2pt/2pt}c ;{2pt/2pt}c ;{2pt/2pt}c  ;{25pt/2pt}c
 ;{2pt/2pt}c ;{2pt/2pt}c ;{25pt/2pt}c  
 ;{2pt/2pt}c ;{2pt/2pt}c ;{2pt/2pt}}
\hdashline[2pt/2pt]
  &  &  x &  & y &   & z & & \\
  \hdashline[2pt/2pt]
 b_{1,4}^* &  b_{2,4}^* &  b_{3,4}^* &  &  &   & & & \\
\hdashline[2pt/2pt]
\end{array}
\right]
\end{equation}
which give $b_{3,2}^* - x$, $b_{3,3}^* - x$, and $b_{3,4}^* - x$, respectively. These terms can be concisely written using $\mathbf{r}_3^* =  [b_{3,1}^*, b_{3,2}^*, b_{3,3}^*, b_{3,4}^*]$ and $\mathbf{x} = [x,x,x,x]$ as
\begin{equation}
    \mathbf{r}_3^* - \mathbf{x}. 
\end{equation} 
Overall, we will have three (in general $m$) expressions: 
\begin{align}
\begin{split}
& \mathbf{r}_3^* - \mathbf{x}\\
& \mathbf{r}_2^* - \mathbf{y} \\ 
& \mathbf{r}_1^* - \mathbf{z}.
\end{split}
\label{eq:rowright1}
\end{align}
For orthogonality with the even multiplicity, the respective rows of Eq. \eqref{eq:rowleft1} must equate with Eq. \eqref{eq:rowright1} as follows:
\begin{align}
\begin{split}
&\mathbf{r}_1 - \mathbf{u} = \mathbf{r}_3^* - \mathbf{x}\\
&\mathbf{r}_2 - \mathbf{u} = \mathbf{r}_2^* - \mathbf{y}\\
&\mathbf{r}_3 - \mathbf{u} = \mathbf{r}_1^* - \mathbf{z}.
\end{split}
\end{align}

We can substitute for $\mathbf{x} = \mathbf{y} = \mathbf{z} = \mathbf{0}$ as the logical codeword we assumed can be re-expressed in its canonical block representation by row-column permutations. Similarly, we can reorder the rows and columns of the base matrix to have zeros in the topmost row and leftmost column to have zeros. Since the goal is to find generic expressions for the base matrix entries, we use the first approach and have the following expressions: 
\begin{equation}
\begin{split}
     &\mathbf{r}_1  - \mathbf{u}  = \mathbf{r}_3^*\\
     &\mathbf{r}_2  - \mathbf{u}  = \mathbf{r}_2^*\\
     &\mathbf{r}_3  - \mathbf{u}  = \mathbf{r}_1^*.
\end{split}   
\end{equation}
Solving for the expressions using the relation for conjugate entries as $b_i^* = L-b_i, \mod L$, we have the \emph{Row Partition Constraint}:
\begin{equation}
\mathbf{u} = \mathbf{r}_1 + \mathbf{r}_3 = 2*\mathbf{r}_2.
\label{eq:RowConstraint}
\end{equation}

\noindent\textit{Example 4:} For the row partition constraint in Eq. \eqref{eq:RowConstraint}, let us take a classical code with circulant size $L=10$ and a base matrix of parity check as follows:
\begin{equation}
B=\begin{bmatrix}
0&0&0&0\\
0&7&1&4\\
0&4&2&8
\end{bmatrix}.
\label{eq:B_3_4_rowConstraintExample}
\end{equation}
This base matrix satisfies the row partition constraint wherein $\mathbf{r}_1 + \mathbf{r}_3 = 2*\mathbf{r}_2 \mod L = [0, 4, 2, 8].$ Hence, for the LP-QLDPC code obtained from the base matrix, we obtain logical codewords with Hamming weight = $m+n$ represented in their canonical block form as
\begin{equation}
    \label{eq:RowConstraintExample}
\bm{\ell} = \left(
\begin{array}{c|c}
\arrayrulecolor{gray!50}
\left[
 \begin{array}
 {;{2pt/2pt}c ;{2pt/2pt}c ;{2pt/2pt}c ;{2pt/2pt}c ;{25pt/2pt}c
 ;{2pt/2pt}c ;{2pt/2pt}c ;{2pt/2pt}c ;{25pt/2pt}c 
 ;{2pt/2pt}c ;{2pt/2pt}c ;{2pt/2pt}c  ;{25pt/2pt}c  
 ;{2pt/2pt}c  ;{2pt/2pt}c  ;{2pt/2pt}c  ;{2pt/2pt}}
\hdashline[2pt/2pt]
0 &   &   &  &  
  & 4 &   &  & 
  &   & 2 &  &  
  &   &   &  8 \\
\hdashline[2pt/2pt]
\end{array}
\right] 
& 
\left[
 \begin{array}{;{2pt/2pt}c ;{2pt/2pt}c ;{2pt/2pt}c ;{25pt/2pt}c 
 ;{2pt/2pt}c ;{2pt/2pt}c ;{25pt/2pt}c 
 ;{2pt/2pt}c ;{2pt/2pt}c ;{2pt/2pt}c}
\hdashline[2pt/2pt]
   &  &  0  & 
   &  0  &  & 
   0  &  &  \\
\hdashline[2pt/2pt]
\end{array}
\right]
\end{array}
\right).
\end{equation}

The question of whether the logical codeword structure we took is generic enough is relevant. The only assumption we started with is that the integer entries must have to be in different blocks and must span all the blocks in order to result in a logical codeword (and not a stabilizer codeword) of the lowest weight possible. The choice of positions of non-zero entries in the logical codeword can indeed be reversed - in the code and transpose side. Since we started with a symmetric LP-QLDPC code, we also have the low-weight codeword that can be expressed based on the partitioning of its blocks as $\bm{\ell} = \left(
\begin{array}{c|c}
\bm{\ell}_{\mathrm{C}} & \bm{\ell}_{\mathrm{T}}
\end{array}
\right),$ where 
\begin{equation}
\bm{\ell}_\mathrm{C}=\arrayrulecolor{gray!50}
\left[
 \begin{array}{;{2pt/2pt}c ;{2pt/2pt}c ;{2pt/2pt}c ;{2pt/2pt}c ;{25pt/2pt}c
 ;{2pt/2pt}c ;{2pt/2pt}c ;{2pt/2pt}c ;{25pt/2pt}c 
 ;{2pt/2pt}c ;{2pt/2pt}c ;{2pt/2pt}c  ;{25pt/2pt}c  
 ;{2pt/2pt}c  ;{2pt/2pt}c  ;{2pt/2pt}c  ;{2pt/2pt}}
\hdashline[2pt/2pt]
   &  &  & a & & & b  & & & c &    & & d  & & &   \\
\hdashline[2pt/2pt]
\end{array}
\right]
\end{equation}
and 
\begin{equation}
\bm{\ell}_\mathrm{T}=\arrayrulecolor{gray!50}
\left[
 \begin{array}{;{2pt/2pt}c ;{2pt/2pt}c ;{2pt/2pt}c ;{25pt/2pt}c 
 ;{2pt/2pt}c ;{2pt/2pt}c ;{25pt/2pt}c 
 ;{2pt/2pt}c ;{2pt/2pt}c ;{2pt/2pt}c}
\hdashline[2pt/2pt]
 x  &  &  &  & y  &  &  &   & z  \\
\hdashline[2pt/2pt]
\end{array}
\right].
\end{equation}
Based on this representation, let us denote the vector corresponding to the elements $x, y,$ and $z$ as a column vector of length $3$ (in general $m$).
\begin{equation}
\label{eq:vectorv}
    \mathbf{v} = [x, y, z]^\mathsf{T}.
\end{equation}
Also, let the vectors $\mathbf{a},\mathbf{b},$ $\mathbf{c},$ and $\mathbf{d}$ represent the column vectors $[a,a,a]^\mathsf{T}$, $[b,b,b]^\mathsf{T}$, $[c,c,c]^\mathsf{T}$, and $[d,d,d]^\mathsf{T}$, respectively. 

We check the orthogonality constraint by taking the minus operation (modulo-$L$) of `code' part $\bm{\ell}_\mathrm{C}$ with the rows on the left part of the LP-QLDPC code of $B_{\mathrm{Z}}$ matrix:= ${I}_4 \otimes B$. 
The first row of ${I}_4 \otimes B$ is $$\arrayrulecolor{gray!50}
\left[
 \begin{array}{
 ;{2pt/2pt}c ;{2pt/2pt}c ;{2pt/2pt}c ;{2pt/2pt}c ;{25pt/2pt}c
 ;{2pt/2pt}c ;{2pt/2pt}c ;{2pt/2pt}c ;{25pt/2pt}c 
 ;{2pt/2pt}c ;{2pt/2pt}c ;{2pt/2pt}c  ;{25pt/2pt}c  
 ;{2pt/2pt}c  ;{2pt/2pt}c  ;{2pt/2pt}c  ;{2pt/2pt}}
\hdashline[2pt/2pt]
b_{1,1} & b_{1,2} & b_{1,3}  & b_{1,4}  & & & & & & &  & & & & &  \\
\hdashline[2pt/2pt]
\end{array}
\right].$$ 
So, we have the two row vectors as the following:
\begin{equation}
\arrayrulecolor{gray!50}
\left[
 \begin{array}{;{2pt/2pt}c ;{2pt/2pt}c ;{2pt/2pt}c ;{2pt/2pt}c ;{25pt/2pt}c
 ;{2pt/2pt}c ;{2pt/2pt}c ;{2pt/2pt}c ;{25pt/2pt}c 
 ;{2pt/2pt}c ;{2pt/2pt}c ;{2pt/2pt}c  ;{25pt/2pt}c  
 ;{2pt/2pt}c  ;{2pt/2pt}c  ;{2pt/2pt}c  ;{2pt/2pt}}
\hdashline[2pt/2pt]
   &    &    &  a &
   &    &  b &    &
   & c  &    &    &
 d &    &    &    \\
\hdashline[2pt/2pt]
b_{1,1} & b_{1,2} & b_{1,3}  & b_{1,4}  & & & & & & &  & & & & &  \\
\hdashline[2pt/2pt]
\hdashline[2pt/2pt]
\end{array}
\right].
\end{equation}
Observe that only one cell overlaps to give the value $b_{1,4} - a$. The subsequent two rows also give one term each as they are of the form:  
\begin{equation}
\arrayrulecolor{gray!50}
\left[
 \begin{array}{;{2pt/2pt}c ;{2pt/2pt}c ;{2pt/2pt}c ;{2pt/2pt}c ;{25pt/2pt}c
 ;{2pt/2pt}c ;{2pt/2pt}c ;{2pt/2pt}c ;{25pt/2pt}c 
 ;{2pt/2pt}c ;{2pt/2pt}c ;{2pt/2pt}c  ;{25pt/2pt}c  
 ;{2pt/2pt}c  ;{2pt/2pt}c  ;{2pt/2pt}c  ;{2pt/2pt}}
\hdashline[2pt/2pt]
   &    &    &  a &
   &    &  b &    &
   & c  &    &    &
 d &    &    &    \\
\hdashline[2pt/2pt]
b_{2,1} & b_{2,2} & b_{2,3}  & b_{2,4}  & & & & & & &  & & & & &  \\
\hdashline[2pt/2pt]
\end{array}
\right]
\end{equation}

\begin{equation}
\arrayrulecolor{gray!50}
\left[
 \begin{array}{;{2pt/2pt}c ;{2pt/2pt}c ;{2pt/2pt}c ;{2pt/2pt}c ;{25pt/2pt}c
 ;{2pt/2pt}c ;{2pt/2pt}c ;{2pt/2pt}c ;{25pt/2pt}c 
 ;{2pt/2pt}c ;{2pt/2pt}c ;{2pt/2pt}c  ;{25pt/2pt}c  
 ;{2pt/2pt}c  ;{2pt/2pt}c  ;{2pt/2pt}c  ;{2pt/2pt}}
\hdashline[2pt/2pt]
 \hdashline[2pt/2pt]
   &    &    &  a &
   &    &  b &    &
   & c  &    &    &
 d &    &    &    \\
\hdashline[2pt/2pt]
b_{3,1} & b_{3,2} & b_{3,3}  & b_{3,4}  & & & & & & &  & & & & &  \\
\hdashline[2pt/2pt]
\hdashline[2pt/2pt]
\end{array}
\right]
\end{equation}
and
\begin{equation}
\arrayrulecolor{gray!50}
\left[
 \begin{array}{;{2pt/2pt}c ;{2pt/2pt}c ;{2pt/2pt}c ;{2pt/2pt}c ;{25pt/2pt}c
 ;{2pt/2pt}c ;{2pt/2pt}c ;{2pt/2pt}c ;{25pt/2pt}c 
 ;{2pt/2pt}c ;{2pt/2pt}c ;{2pt/2pt}c  ;{25pt/2pt}c  
 ;{2pt/2pt}c  ;{2pt/2pt}c  ;{2pt/2pt}c  ;{2pt/2pt}}
\hdashline[2pt/2pt]
\hdashline[2pt/2pt]
   &    &    &  a &
   &    &  b &    &
   & c  &    &    &
 d &    &    &    \\
\hdashline[2pt/2pt]
 &  &  & b_{1,1} &  &  & & b_{1,2}  &  &  &   & b_{1,3} & &  &  & b_{1,4} \\
\hdashline[2pt/2pt]
\hdashline[2pt/2pt]
\end{array}
\right].
\end{equation}
Overall, the minus operation results in the following 3 (or in general $m$) terms: 
$b_{1,4} - a$, $b_{2,4} - a$, and $b_{3,4} - a$. Utilizing the column vector representation, we have: $\mathbf{c}_4 - \mathbf{a}$, where $\mathbf{a} = [a,a,a]^{\mathsf{T}}.$  As in the case of derivation of the row partition constraint, the difference of $\bm{\ell}_\mathrm{C}$ with the rows on the left part of the LP-QLDPC code of $B_Z$ matrix:= $\mathbf{I}_4 \otimes B$ gives 4 (in general $n$) terms that can be expressed as follows:
\begin{align}
\begin{split}
&\mathbf{c}_4 - \mathbf{a}\\
&\mathbf{c}_3 - \mathbf{b}\\
&\mathbf{c}_2 - \mathbf{c}\\ 
&\mathbf{c}_1 - \mathbf{d}.
\end{split}
\label{eq:ColEqLeft}
\end{align}
Now, let us take the minus operation of $\bm{\ell}_\mathrm{T}$ with the rows on the transpose part of the $B_Z$ matrix:= $B^* \otimes \mathbf{I}_3$. Recall that we have 
\begin{equation}
\bm{\ell}_\mathrm{T}=\arrayrulecolor{gray!50}
\left[
 \begin{array}{;{2pt/2pt}c ;{2pt/2pt}c ;{2pt/2pt}c  ;{25pt/2pt}c
 ;{2pt/2pt}c ;{2pt/2pt}c ;{25pt/2pt}c  
 ;{2pt/2pt}c ;{2pt/2pt}c ;{2pt/2pt}}
\hdashline[2pt/2pt]
 x &   &    &
   & y &    & 
   &   & z \\
\hdashline[2pt/2pt]
\end{array}
\right].
\end{equation}
To perform the minus operation with the first row of  $B^* \otimes {I}_3$, that is $
\arrayrulecolor{gray!50}
\left[
 \begin{array}{;{2pt/2pt}c ;{2pt/2pt}c ;{2pt/2pt}c  ;{25pt/2pt}c
 ;{2pt/2pt}c ;{2pt/2pt}c ;{25pt/2pt}c  
 ;{2pt/2pt}c ;{2pt/2pt}c ;{2pt/2pt}}
\hdashline[2pt/2pt]
 b_{1,1}^* &            &           & 
b_{2,1}^*  &            &           &
 b_{3,1}^* &            &           \\
\hdashline[2pt/2pt]
\end{array}
\right]$, we have 
\begin{equation}
\arrayrulecolor{gray!50}
\left[
 \begin{array}{;{2pt/2pt}c ;{2pt/2pt}c ;{2pt/2pt}c  ;{25pt/2pt}c
 ;{2pt/2pt}c ;{2pt/2pt}c ;{25pt/2pt}c  
 ;{2pt/2pt}c ;{2pt/2pt}c ;{2pt/2pt}}
\hdashline[2pt/2pt]
 x &   &    &
   & y &    & 
   &   & z \\
\hdashline[2pt/2pt]
 b_{1,1}^* &            &           & 
b_{2,1}^*  &            &           &
 b_{3,1}^* &            &           \\
\hdashline[2pt/2pt]
\end{array}
\right].
\end{equation}
The overlapping term in the first cell gives $b_{1,1}^* - x$. 
Similarly, we can express the next two rows as follows: 
\begin{equation}
\arrayrulecolor{gray!50}
\left[
 \begin{array}{;{2pt/2pt}c ;{2pt/2pt}c ;{2pt/2pt}c  ;{25pt/2pt}c
 ;{2pt/2pt}c ;{2pt/2pt}c ;{25pt/2pt}c  
 ;{2pt/2pt}c ;{2pt/2pt}c ;{2pt/2pt}}
\hdashline[2pt/2pt]
 x &   &    &
   & y &    & 
   &   & z \\
\hdashline[2pt/2pt]
 &   b_{1,1}^*     &           & 
 &   b_{2,1}^*     &           &
 &   b_{3,1}^*     &           \\
\hdashline[2pt/2pt]
\end{array}
\right]
\end{equation}

\begin{equation}
\arrayrulecolor{gray!50}
\left[
 \begin{array}{;{2pt/2pt}c ;{2pt/2pt}c ;{2pt/2pt}c  ;{25pt/2pt}c
 ;{2pt/2pt}c ;{2pt/2pt}c ;{25pt/2pt}c  
 ;{2pt/2pt}c ;{2pt/2pt}c ;{2pt/2pt}}
\hdashline[2pt/2pt]
 x &   &    &
   & y &    & 
   &   & z \\
\hdashline[2pt/2pt]
 &            &     b_{1,1}^*     & 
 &            &     b_{2,1}^*     &
 &            &     b_{3,1}^*       \\
\hdashline[2pt/2pt]
\end{array}
\right].
\end{equation}
We obtain $b_{1,1}^* -x$, 
$b_{2,1}^* -y$, and $b_{3,1}^* -z$ as the result and we can concisely write using the column representation in Eq. \eqref{eq:col_cjConj} and \eqref{eq:vectorv} as $\mathbf{c}_1^* - \mathbf{v}$. Considering the difference of $\ell_\mathrm{C}$ with all rows, we get the following simplified expressions. 
\begin{align}
\begin{split}
&\mathbf{c}_1^* - \mathbf{v}\\
&\mathbf{c}_2^* - \mathbf{v}\\ 
&\mathbf{c}_3^* - \mathbf{v}\\ 
&\mathbf{c}_4^* - \mathbf{v}. 
\end{split}
\label{eq:ColEqRight}
\end{align}

Using the same argument that this codeword can be expressed in a canonical way, we obtain the linear equations to be satisfied for $\bm{\ell}$ to be orthogonal as: 
\begin{align}
\begin{split}
    &\mathbf{c}_4 = \mathbf{c}_1^* - \mathbf{v} \\
    &\mathbf{c}_3 = \mathbf{c}_2^* - \mathbf{v}\\
    &\mathbf{c}_2 = \mathbf{c}_3^* - \mathbf{v}\\
    &\mathbf{c}_1 = \mathbf{c}_4^*- \mathbf{v}.
\end{split}
\end{align}
Solving for the expressions using the relation for conjugate entries as $b_i^* = L-b_i$, $\mod L$, we have the \emph{Column Partition Constraint}:
\begin{equation}
\mathbf{v}^* = \mathbf{c}_1 + \mathbf{c}_4 = \mathbf{c}_2 + \mathbf{c}_3.
\label{eq:ColumnConstraint}
\end{equation}

\noindent\textit{Example 2 Revisited: } To illustrate the column partition constraint in Eq. \eqref{eq:ColumnConstraint}, let us revisit our Example 2 with the classical code $\mathcal{C}: [104,30,14]$, and the base matrix of PCM with circulant size $L=26$,
\begin{equation}
B=\begin{bmatrix}
0&0&0&0\\
0&6&4&10\\
0&8&14&22
\end{bmatrix}.
\end{equation}
This base matrix satisfies the column partition constraint wherein $\mathbf{c}_1 + \mathbf{c}_4 = \mathbf{c}_2 + \mathbf{c}_3 \mod 26 = [0\; 10\; 22]^\mathsf{T}.$ Hence, for the LP-QLDPC code obtained from the base matrix, we obtain logical codewords with Hamming weight = $m+n$ represented in their canonical block form as follows: 
\begin{equation}
\bm{\ell} = \left(
\begin{array}{c|c}
\arrayrulecolor{gray!50}
\left[
 \begin{array}{;{2pt/2pt}c ;{2pt/2pt}c ;{2pt/2pt}c ;{2pt/2pt}c ;{25pt/2pt}c
 ;{2pt/2pt}c ;{2pt/2pt}c ;{2pt/2pt}c ;{25pt/2pt}c 
 ;{2pt/2pt}c ;{2pt/2pt}c ;{2pt/2pt}c  ;{25pt/2pt}c  
 ;{2pt/2pt}c  ;{2pt/2pt}c  ;{2pt/2pt}c  ;{2pt/2pt}}
\hdashline[2pt/2pt]
   &  &  & 0 & & & 0  & & & 0 &    & & 0  & & &   \\
\hdashline[2pt/2pt]
\end{array}
\right] 
&
\left[
 \begin{array}{;{2pt/2pt}c ;{2pt/2pt}c ;{2pt/2pt}c ;{25pt/2pt}c 
 ;{2pt/2pt}c ;{2pt/2pt}c ;{25pt/2pt}c 
 ;{2pt/2pt}c ;{2pt/2pt}c ;{2pt/2pt}c}
\hdashline[2pt/2pt]
 0  &  &  &  & 16  &  &  &   & 4  \\
\hdashline[2pt/2pt]
\end{array}
\right]
\end{array}
\right).
\end{equation}

\noindent\textit{Example 4 Continued:}
We looked at the base matrix in Eq. \eqref{eq:B_3_4_rowConstraintExample} as an example for the row partition constraint. However, it also satisfies the column partition constraint in Eq. \eqref{eq:ColumnConstraint} wherein $\mathbf{c}_1 + \mathbf{c}_3 = \mathbf{c}_2 + \mathbf{c}_4 \mod 10 = [0\; 1\; 2]^\mathsf{T}.$ 
Thus, an additional weight-7 logical codeword can be found that is represented in its canonical form as 
\begin{equation}
\bm{\ell} = \left(
\begin{array}{c|c}
\arrayrulecolor{gray!50}
\left[
 \begin{array}{;{2pt/2pt}c ;{2pt/2pt}c ;{2pt/2pt}c ;{2pt/2pt}c ;{25pt/2pt}c
 ;{2pt/2pt}c ;{2pt/2pt}c ;{2pt/2pt}c ;{25pt/2pt}c 
 ;{2pt/2pt}c ;{2pt/2pt}c ;{2pt/2pt}c  ;{25pt/2pt}c  
 ;{2pt/2pt}c  ;{2pt/2pt}c  ;{2pt/2pt}c  ;{2pt/2pt}}
\hdashline[2pt/2pt]
   &  &  & 0 & & & 0  & & & 0 &    & & 0  & & &   \\
\hdashline[2pt/2pt]
\end{array}
\right] 
&
\left[
 \begin{array}{;{2pt/2pt}c ;{2pt/2pt}c ;{2pt/2pt}c ;{25pt/2pt}c 
 ;{2pt/2pt}c ;{2pt/2pt}c ;{25pt/2pt}c 
 ;{2pt/2pt}c ;{2pt/2pt}c ;{2pt/2pt}c}
\hdashline[2pt/2pt]
 0  &  &  &  & 9  &  &  &   & 8  \\
\hdashline[2pt/2pt]
\end{array}
\right]
\end{array}
\right).
\end{equation}

We have essentially expressed the row and column partition constraints that will create low-weight codewords using an example base matrix. In the next section, we will formalize it for the general case as well as discuss special implications.

\subsection{Generalized Row Column Partition Constraints}
Consider a type-1 quasi-cyclic LDPC parity check base matrix $B$ with $m$ rows and $n$ columns and with circulant size $L$. 
Let $\mathbf{c}_j$ and $\mathbf{r}_i$ represent the $j$-th column and $i$-th row of $B$, respectively. 
Suppose $n$ (or $m$) is even, let $\mathcal{P}$ (or $\mathcal{R}$) be a partition of the set of columns (or rows) into $n/2$ (or $m/2$) pairs, where each pair consists of distinct columns (or rows)
\begin{equation}
\mathcal{P} = \{ (\mathbf{c}_{i_1}, \mathbf{c}_{j_1}), (\mathbf{c}_{i_2}, \mathbf{c}_{j_2}), \ldots, (\mathbf{c}_{i_{n/2}}, \mathbf{c}_{j_{n/2}}) \},
\label{eq:evencolPartition}
\end{equation}
where $\{i_1, j_1, i_2, j_2, \dots, i_{n/2}, j_{n/2} \}$ is any permutation of the index set $\{1, 2, \dots, n\}$, meaning each index appears exactly once in the collection.   
For the partition of rows, we have 
\begin{equation}
\mathcal{R} = \{ (\mathbf{r}_{i_1}, \mathbf{r}_{j_1}), (\mathbf{r}_{i_2}, \mathbf{r}_{j_2}), \ldots, (\mathbf{r}_{i_{m/2}}, \mathbf{r}_{j_{m/2}}) \}, 
\label{eq:evenrowPartition}
\end{equation}
where $\{i_1, j_1, i_2, j_2, \dots, i_{m/2}, j_{m/2} \}$ is any permutation of the index set $\{1, 2, \dots, m\}$.\\
For the case when the number of columns $n$ (or rows $m$) is odd, let $\mathcal{P}$ (or $\mathcal{R}$) be a partition of the set of columns (or rows) into pairs, including a column $\mathbf{c}_{i_{(n+1)/2}}$ (or row $\mathbf{r}_{i_{(m+1)/2}}$) that is paired with itself:
\begin{equation}
\mathcal{P} = \{ (\mathbf{c}_{i_1}, \mathbf{c}_{j_1}), (\mathbf{c}_{i_2}, \mathbf{c}_{j_2}), \ldots, (\mathbf{c}_{i_{(n-1)/2}}, \mathbf{c}_{j_{(n-1)/2}}), (\mathbf{c}_{i_{(n+1)/2}}, \mathbf{c}_{i_{(n+1)/2}}) \}.
\label{eq:oddcolPartition}
\end{equation}
\begin{equation}    
\mathcal{R} = \{ (\mathbf{r}_{i_1}, \mathbf{r}_{j_1}), (\mathbf{r}_{i_2}, \mathbf{r}_{j_2}), \ldots, (\mathbf{r}_{i_{(m-1)/2}}, \mathbf{r}_{j_{(m-1)/2}}), (\mathbf{r}_{i_{(m+1)/2}}, \mathbf{r}_{i_{(m+1)/2}}) \}.
\label{eq:oddrowPartition}
\end{equation}
Here, $\{i_1, j_1, \dots, i_{(n-1)/2}, j_{(n-1)/2}, i_{(n+1)/2} \}$ is any permutation of the index set $\{1, \dots, n\}$, and similarly for the row indices as well. 
For the orthogonality of codewords in the LP-QLDPC code, we need to check if the columns (or rows) of the base matrix $B$ can be paired such that the modulo-$L$ sum of each of the pairs is equal. As mentioned before, all operations are modulo-$L$, applied element-wise to the columns (or rows). 
Now, we define the row/column partitioning constraint (RCPC) as follows:\\

\noindent \textbf{Column Partitioning Constraint:}
\begin{definition}[Even Number of Columns]
\label{def_even_column}
    Given a base matrix with $n \in 2\mathbb{Z}$ columns and distinct indices $1 \le h,i,j,k \le n$, the column partitioning constraint is satisfied if there exists a partition $\mathcal{P}$ as in Eq. \eqref{eq:evencolPartition} such that
\begin{equation}
         \forall (\mathbf{c}_i, \mathbf{c}_j) \text{ and } \forall (\mathbf{c}_h, \mathbf{c}_k)  \in  \mathcal{P}, \, \mathbf{c}_i + \mathbf{c}_j = \mathbf{c}_h + \mathbf{c}_k.
    \end{equation}
    \label{eq:ColPartitionEvenCase}
\end{definition}

\begin{definition}[Odd Number of Columns]
\label{def_odd_column}
    Given a base matrix with $n \in 2\mathbb{Z}+1$ columns and distinct indices $1 \le h,i,j,k \le n$, and $k \neq i, k \neq j$ being the column index not corresponding to a pair, the column partitioning constraint is satisfied if there exists a partition $\mathcal{P}$ as in Eq. \eqref{eq:oddcolPartition} such that
 \begin{align}
    &\left\{
    \begin{array}{l}
        \forall (\mathbf{c}_i, \mathbf{c}_j)  \text{ and } \forall (\mathbf{c}_h, \mathbf{c}_k) \in  \mathcal{P}, \quad  \mathbf{c}_i + \mathbf{c}_j = \mathbf{c}_h + \mathbf{c}_k,\\
        \forall (\mathbf{c}_i, \mathbf{c}_j)  \text{ and } \;\; (\mathbf{c}_k, \mathbf{c}_k) \in  \mathcal{P}, \quad  \mathbf{c}_i + \mathbf{c}_j = \mathbf{c}_k + \mathbf{c}_k.
    \end{array}
    \right.
    \label{eq:ColPartitionOddCase}
\end{align}

\end{definition}

\noindent \textbf{Row Partitioning Constraint:}
\begin{definition}[Even Number of Rows]
\label{def_even_row}
    Given a base matrix with $m \in 2\mathbb{Z}$ rows and distinct indices $1 \le h,i,j,k \le m$, the row partitioning constraint is satisfied if there exists a partition $\mathcal{R}$ as in Eq. \eqref{eq:evenrowPartition} such that
    \begin{equation}
         \forall (\mathbf{r}_i, \mathbf{r}_j) \text{ and } \forall (\mathbf{r}_h, \mathbf{r}_k) \in  \mathcal{R}, \quad \mathbf{r}_i + \mathbf{r}_j =  \mathbf{r}_h + \mathbf{r}_k.
         \label{eq:RowPartitionEvenCase}
    \end{equation}
\end{definition}

\begin{definition}[Odd Number of Rows]
\label{def_odd_row}
    Given a base matrix with $m \in 2\mathbb{Z}+1$ rows, and distinct indices $1 \le h,i,j,k \le m$ and $k\neq i, k \neq j$ being the row index not corresponding to a pair, the row partitioning constraint is satisfied if there exists a partition $\mathcal{R}$ as in Eq. \eqref{eq:oddrowPartition} such that
    \begin{align}
    &\left\{
    \begin{array}{l}
        \forall (\mathbf{r}_i, \mathbf{r}_j)  \text{ and } \forall (\mathbf{r}_h, \mathbf{r}_k) \in  \mathcal{R}, \quad  \mathbf{r}_i + \mathbf{r}_j = \mathbf{r}_h + \mathbf{r}_k,\\
        \forall (\mathbf{r}_i, \mathbf{r}_j)  \text{ and } \;\; (\mathbf{r}_k, \mathbf{r}_k) \in  \mathcal{R}, \quad  \mathbf{r}_i + \mathbf{r}_j = \mathbf{r}_k + \mathbf{r}_k.
    \end{array}
    \right.
    \label{eq:RowPartitionOddCase}
\end{align}
\end{definition}

With Definitions 2-5, the existence of a minimum weight codeword whose Hamming weight is equal to the stabilizer weight can be characterized with the RCPC in the following theorem.\\

\begin{te}
    Given a quasi-cyclic base matrix $B$ with minimum distance $d_{\text{min}}^{\mathcal{C}}$ and of size $m \times n$ of an LP-QLDPC code, if there exists partitioning of rows or columns such that the row/column partitioning constraint (RCPC) is satisfied, then the LP-QLDPC code has a minimum distance $d_{\text{min}}^{\mathcal{Q}} = \min(d_{\text{min}}^{\mathcal{C}},m + n)$.  
\end{te}
\begin{proof}
    We present the proof sketch for the row partition constraint, even case, and omit the other cases due to space constraints, as they follow similarly. A full proof will be provided in future publication.
    Let $\mathbf{r}_i$, $1 \le i \le m$ denote the rows of the base matrix $B$.
    Let us assume that the row partition constraint is satisfied based on Definition \ref{def_even_row}. 
    This implies that there exists a row partition $\mathcal{R}$ as in Eq. \eqref{eq:RowPartitionEvenCase}. One such partition follows the property: 
    \begin{equation}
        \mathbf{r}_1 + \mathbf{r}_m = \mathbf{r}_2 + \mathbf{r}_{m-1} = \ldots = \mathbf{r}_{m/2-1} + \mathbf{r}_{m/2 +1}  = (u_1,u_2,\ldots, u_n).
    \label{eq:SpecialPartition}
    \end{equation} 
    We now prove that there exists a logical codeword of the following block form:
    \begin{equation}
    \bm{\ell} = \left(
    \begin{array}{c|c}
    \arrayrulecolor{gray!50}
    \left[
    \begin{array}
 {;{2pt/2pt}c ;{2pt/2pt}c ;{2pt/2pt}c ;{2pt/2pt}c ;{25pt/2pt}c
 ;{2pt/2pt}c ;{2pt/2pt}c ;{2pt/2pt}c ;{25pt/2pt}c 
 ;{2pt/2pt}c ;{2pt/2pt}c ;{2pt/2pt}c  ;{25pt/2pt}c  
 ;{2pt/2pt}c  ;{2pt/2pt}c  ;{2pt/2pt}c  ;{2pt/2pt}}
\hdashline[2pt/2pt]
u_{1} &   &   &  &  
  & u_{2} &   &  & 
  &   & \ldots &  &  
  &   &   &  u_{n} \\
\hdashline[2pt/2pt]
\end{array}
\right] 
& 
\left[
 \begin{array}
 { 
 ;{2pt/2pt}c
 ;{2pt/2pt}c 
 ;{2pt/2pt}c
 ;{2pt/2pt}c 
 ;{25pt/2pt}c 
 ;{2pt/2pt}c
 ;{2pt/2pt}c
 ;{2pt/2pt}c
 ;{25pt/2pt}c
 ;{25pt/2pt}c 
 ;{2pt/2pt}c 
 ;{2pt/2pt}c 
 ;{2pt/2pt}c 
;{25pt/2pt}c }
\hdashline[2pt/2pt]
 & & & 0 &
 & & 0 & &
 \ldots 
 & 0 & & & \\
\hdashline[2pt/2pt]
\end{array}
\right]
\end{array}
\right).
\label{eq:GenericRowconstraint}
\end{equation}

In Eq. \eqref{eq:GenericRowconstraint}, each of the $n$ blocks in the code part is composed of $n$ cells, and the $m$ blocks in the transpose part have $m$ cells each. 
First, this vector is not a stabilizer code word - both of $B_\mathrm{X}$  and of $B_\mathrm{Z}$.  
We now prove that it is indeed orthogonal to the base matrices of the LP-QLDPC code. For orthogonality, we check the difference of the corresponding vectors cell by cell for even multiplicity condition in the same way as proved in Theorem \ref{Th:EvenMu}.  
Take any row in $B_\mathrm{X}$ or $B_\mathrm{Z}$; Only two integer terms exist in the difference between the logical codeword $\bm{\ell}$ and any row of the base matrices - one in the code and one in the transpose part. 
The difference terms are either both 0 when the overlapping terms are identical or they are equal as a consequence of the row constraint expressions in Eq.~\eqref{eq:SpecialPartition}. 
One can verify that the even overlap condition is satisfied with both $B_\mathrm{X}$ and $B_\mathrm{Z}$.
Thus, the logical codeword in Eq. \eqref{eq:GenericRowconstraint} is orthogonal to both the base matrices. 
Furthermore, the weight of the codeword is $m+n$. This proves the existence of logical codewords of Hamming weight  $m+n$. 
If the base code has a minimum distance $d_{\text{min}}^{\mathcal{C}} < m+n$, those minimum weight codewords determine the minimum distance. Otherwise, logical codewords of Hamming weight $m+n$ result in $d_{\text{min}}^{\mathcal{Q}} = m+n$. Thus, we prove that the quantum minimum distance is the minimum of $d_{\text{min}}^{\mathcal{C}}$ and $m+n$. We have, 
$d_{\text{min}}^{\mathcal{Q}} = \min( d_{\text{min}}^{\mathcal{C}}, m+n)$. This concludes the proof.  
\end{proof}

\textit{Corollary 1:} 
For a base matrix with $m = 2$ (or $n =2$), the RCPC is always satisfied restricting the minimum distance to be upper-bounded at $n+2$ (or $m+2$) regardless of the classical code chosen. 
\begin{proof}
    Let us again prove for the row constraint case, where such a base matrix $B$ of size $2 \times n$ with circulant size $L$ can be expressed in general with the first row and column of integer circulants shifted to zeros as
\begin{equation}
B=\begin{bmatrix}
0 & 0 & \ldots & 0\\
0 & b_{2,2} & \ldots & b_{2,n}
\end{bmatrix}.
\label{eq: B_2_3_b}
\end{equation}
It is easy to see that one can always find partitioning of rows of Eq. \eqref{eq: B_2_3_b} such that $\mathbf{r_1} + \mathbf{r_2} = \mathbf{u}$, where $\mathbf{u} = (0 \; b_{2,2} \;  \ldots \;  b_{2,n})$. Thus, there always exists logical codewords of Hamming weight $n + m$, where $m = 2$, with a canonical block form: $$\bm{\ell} = \left(
\begin{array}{c|c}
\arrayrulecolor{gray!50}
\left[
 \begin{array}
 {;{2pt/2pt}c ;{2pt/2pt}c ;{2pt/2pt}c ;{2pt/2pt}c ;{25pt/2pt}c
 ;{2pt/2pt}c ;{2pt/2pt}c ;{2pt/2pt}c ;{25pt/2pt}c 
 ;{2pt/2pt}c ;{2pt/2pt}c ;{2pt/2pt}c  ;{25pt/2pt}c  
 ;{2pt/2pt}c  ;{2pt/2pt}c  ;{2pt/2pt}c  ;{2pt/2pt}}
\hdashline[2pt/2pt]
0  &   &   &  &  
  & b_{2,2} &   &  & 
  &   & \ldots &  &  
  &   &   &  b_{2,n} \\
\hdashline[2pt/2pt]
\end{array}
\right] 
& 
\left[
 \begin{array}
 { ;{2pt/2pt}c;{2pt/2pt}c ;{25pt/2pt}c 
 ;{2pt/2pt}c ;{25pt/2pt}c 
 ;{2pt/2pt}c }
\hdashline[2pt/2pt]
     &  0 &  0  &  \\
\hdashline[2pt/2pt]
\end{array}
\right]
\end{array}
\right).$$

\end{proof}

\noindent\textit{Example 1 Revisited:} $\mathcal{C}:[21,8,6]$ LDPC code with circulant size $L=7$ and base matrix $B=\begin{bmatrix}
0&0&0\\
0&1&3
\end{bmatrix}$ satisfies RCPC condition as $m=2$. The minimum distance of the constructed LP-QLDPC code is always upper-bounded by $n+2 = 5$.

\section{DESIGN OF DEGENERATE LP-QLDPC CODES}
\label{sec_CodeDesign}
\subsection{Impact on Code Performance}
As an illustration of the impact of non-degenerate LP-QLDPC codes, we compare the decoding performances for two LP-QLDPC codes with the same code length ($N=600$), base matrix dimensions ($m = 3,n = 4$), minimum distance of the base code ($d_{\text{min}}^\mathcal{C} =20$), and girth $g=8$ - only differing in the quantum minimum distance after constructing the LP-QLDPC codes. We chose the circulant size $L=24$ and base matrices to be 
\begin{equation}
B_{d_{\text{min}}^\mathcal{Q} =7}^{(1)}=\begin{bmatrix}
0&0&0&0\\
0&1&2&3\\
0&5&9&14
\end{bmatrix}, 
\quad
B_{d_{\text{min}}^\mathcal{Q} = 20}^{(2)}=\begin{bmatrix}
0&0&0&0\\
0&1&3&7\\
0&9&16&5
\end{bmatrix}.
\end{equation}
The LP-QLDPC code constructed from $B_{d_{\text{min}}^\mathcal{Q} =7}^{(1)}$ is a non-degenerate quantum code as the quantum minimum distance is equal to the weight of stabilizer generators which is $m+n = 7$. As for the LP-QLDPC code constructed from $B_{d_{\text{min}}^\mathcal{Q} =20}^{(2)}$, it is sufficiently degenerate with high $d_{\text{min}}^\mathcal{Q} = 20$ compared to the weight of stabilizer generators i.e., $m+n = 7$.
We compare their decoding performances under code capacity model using a belief propagation decoder followed by ordered statistics decoding \cite{panteleev_degenerate_2021} (order-10) in Fig. \ref{fig:DecoderSim}. The difference in the \emph{waterfall region of the curve} where the degenerate code outperforms the non-degenerate code shows that the decoder can converge to a significant number of degenerate error patterns resulting in a steeper slope at high error rates.  

\begin{figure}
    \centering
    \includegraphics[width=0.5\linewidth]{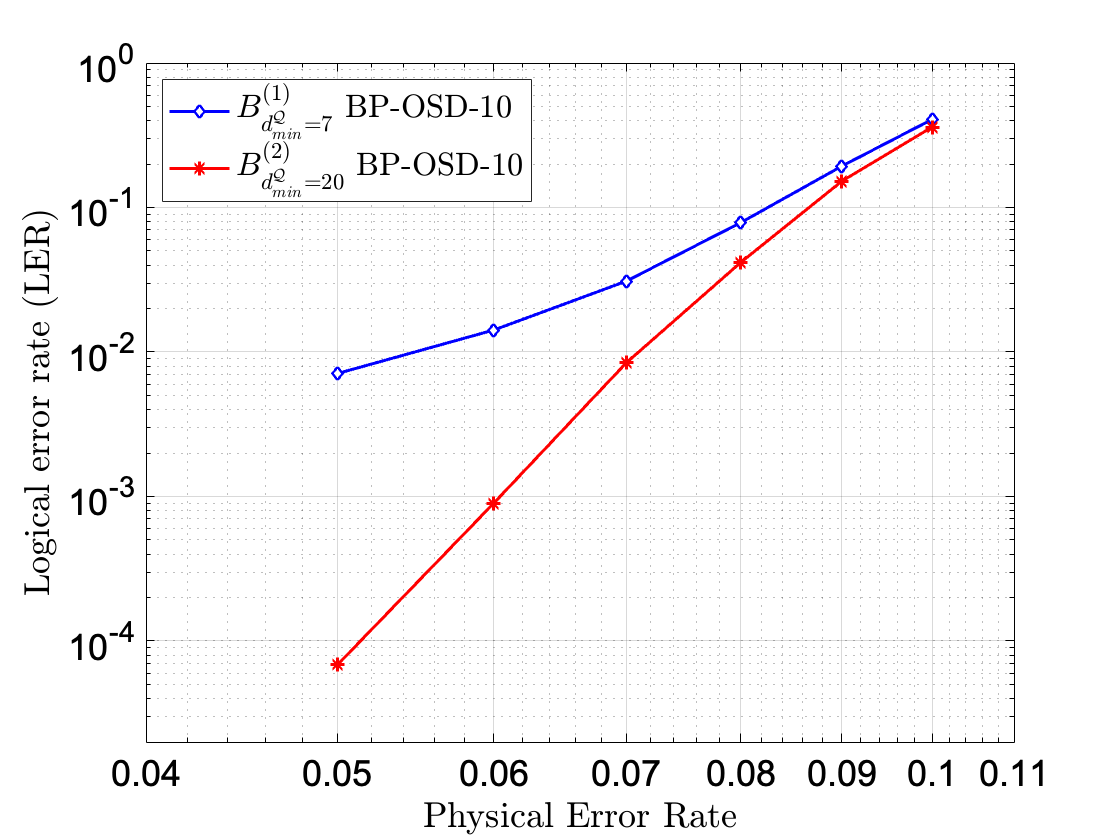}
    \caption{Comparison of decoder performance (BP with OSD-10) for the two LP-QLDPC codes [[600,36,20]] and [[600,36,7]] in the high error rate region. From the drastic change of slopes, it is clear that the choice of non-degenerate code significantly degrades the decoding performance - resulting in numerous logical errors of low weight.}
    \label{fig:DecoderSim}
\end{figure}

\subsection{Degenerate LP-QLDPC Codes} 
The RCPC-based recipe for code construction guarantees degenerate LP-QLDPC codes whose quantum minimum distance is strictly greater than the weight of the stabilizers in the parity check matrix. 
For a given circulant size $L$, using our recipe for LP code construction, we can list all the parity check matrices constructed that have the desired properties.
The first column in Table \ref{tab:m3n4QC_LPCodes} shows the circulant sizes chosen for the different $m = 3, n = 4$ quasi-cyclic base matrices. The second column shows the possible number of base matrices exhaustively searched through - calculated as $L^{m-1 \times n-1}$, as we look only for equivalent matrices with zero circulants in both the first row and first column of base matrix, thus reducing the exhaustive search space from $L^{m \times n}$. The third column gives the number of base matrices with logical codewords of Hamming weight = $m+n$ as they satisfy RCPC. 
The subsequent column lists the number of base matrices such that the minimum distance of the base code is strictly greater than the stabilizer weight. 
Avoiding RCPC satisfied base matrices gives the total number of degenerate LP-QLDPC codes that could be constructed with these parameters in the last column. These numbers are big as we perform an exhaustive search, and many of the base matrices have girth $g= 4$. 
Increasing the circulant size as well as optimizing the minimum distance and girth of the base matrix with other code parameters, we can effectively reduce the number of base matrices to select from. 
For instance, for $L=13$, precisely 2592 of base matrices have $d_{\text{min}}^{\mathcal{Q}} = d_{\text{min}}^{\mathcal{C}} = 14$ and girth $g = 8$. In future work, we plan to develop our framework to incorporate relevant parameter choices of interest, such as specific girth constraints, minimum distance and optimized protographs for practical applications.

\begin{table}[h!]
\centering
\caption{Number of Base matrices with $m=3, n=4$ for different circulant sizes}
\label{tab:m3n4QC_LPCodes}
\begin{tabular}{|c|c|c|c|c|}
\hline
\textbf{Circulant size ($L$)} & \textbf{Exhaustive} & \textbf{RCPC} & $d_{min}^{\mathcal{C}} > 7$ & $d_{min}^{\mathcal{Q}} > 7$ \\ \hline
4  & $4^{6} = 4096$  & 700 & 0 & 0 \\ \hline
5  & 15625    & 1993 & 0& 0 \\ \hline
6  & 46656    & 3876 & 1008 & 720 \\ \hline
7  & 117649   & 7725 & 55440 & 53856 \\ \hline
8  & 262144   & 12628 & 128592 & 124704 \\ \hline
9  & 531441   & 20961 & 189360 & 183744\\ \hline
10 & 1000000  & 31084 & 656784 & 644544 \\ \hline
\end{tabular}
\end{table}

\section{Conclusions and Future Direction}
\label{sec:Conclusion}
In this work, we focused on designing type-1 quasi-cyclic LP-QLDPC codes with guarantees on their minimum distances. Theoretical proof for ensuring degeneracy (higher minimum distance than stabilizer weight) of LP-QLDPC codes is given. This is a key step in our methodology for constructing good quantum codes. The next step is to derive sufficient constraints to guarantee that the codes designed have $d_{\text{min}}^{\mathcal{Q}} = d_{\text{min}}^{\mathcal{C}}$. In future works, we aim to explore combinatorial conditions that remove newly identified problematic structures and extend our analysis to more general protographs. Constructing more structured quasi-cyclic codes that automatically guarantee minimum distance preservation property is another open problem. This can leverage distance bounds and strategies for classical quasi-cyclic protograph LDPC codes. Also, we are interested in investigating the choice of circulant matrix entries that ease hardware constraints. 

\section*{Acknowledgment} B. Vasi\'{c} acknowledges the support of the NSF under grants CCF-2420424, CIF-2106189, NSF-ERC 1941583, CCF-2100013, and CCSS-2052751. 
Authors thank Michele Pacenti for assisting with the preliminary simulation. 
B. Vasi\'{c} and D. Declercq have disclosed an outside interest in Codelucida to the University of Arizona. Conflicts of interest resulting from this interest are being managed by The University of Arizona in accordance with its policies.



\begin{thebibliography}{10}
\providecommand{\url}[1]{#1}
\csname url@samestyle\endcsname
\providecommand{\newblock}{\relax}
\providecommand{\bibinfo}[2]{#2}
\providecommand{\BIBentrySTDinterwordspacing}{\spaceskip=0pt\relax}
\providecommand{\BIBentryALTinterwordstretchfactor}{4}
\providecommand{\BIBentryALTinterwordspacing}{\spaceskip=\fontdimen2\font plus
\BIBentryALTinterwordstretchfactor\fontdimen3\font minus
  \fontdimen4\font\relax}
\providecommand{\BIBforeignlanguage}[2]{{%
\expandafter\ifx\csname l@#1\endcsname\relax
\typeout{** WARNING: IEEEtran.bst: No hyphenation pattern has been}%
\typeout{** loaded for the language `#1'. Using the pattern for}%
\typeout{** the default language instead.}%
\else
\language=\csname l@#1\endcsname
\fi
#2}}
\providecommand{\BIBdecl}{\relax}
\BIBdecl

\bibitem{panteleev_asymptotically_2021}
\BIBentryALTinterwordspacing
P.~Panteleev and G.~Kalachev, ``Asymptotically good quantum and locally
  testable classical {LDPC} codes,'' in \emph{Proc. 54th Ann. ACM Symp. Theory
  of Computing}, 2022, p. 375–388. [Online]. Available:
  \url{https://doi.org/10.1145/3519935.3520017}
\BIBentrySTDinterwordspacing

\bibitem{nature_qldpc_demonstration}
\BIBentryALTinterwordspacing
Q.~Xu, J.~P.~B. Ataides, C.~A. Pattison, N.~Raveendran, D.~Bluvstein, J.~Wurtz,
  B.~Vasi\'c, M.~D. Lukin, L.~Jiang, and H.~Zhou, ``Constant-overhead
  fault-tolerant quantum computation with reconfigurable atom arrays,''
  \emph{Nature Physics}, pp. 1--10, April 2024. [Online]. Available:
  \url{https://doi.org/10.1038/s41567-024-02479-z}
\BIBentrySTDinterwordspacing

\bibitem{bravyi2023_IBM_highthreshold}
\BIBentryALTinterwordspacing
S.~Bravyi, A.~W. Cross, J.~M. Gambetta, D.~Maslov, P.~Rall, and T.~J. Yoder,
  ``High-threshold and low-overhead fault-tolerant quantum memory,''
  \emph{Nature 627}, pp. 778--782, March 2024. [Online]. Available:
  \url{https://doi.org/10.1038/s41586-024-07107-7}
\BIBentrySTDinterwordspacing

\bibitem{IBM_2024_LinearAncillaLogicalMeasurementsBBCode}
A.~Cross, Z.~He, P.~Rall, and T.~Yoder, ``Linear-size ancilla systems for
  logical measurements in {QLDPC} codes,'' \emph{arXiv preprint
  arXiv:2407.18393}, 2024.

\bibitem{IBM_2024_Cross_improvedqldpcsurgerylogical}
\BIBentryALTinterwordspacing
------, ``{Improved QLDPC Surgery: Logical Measurements and Bridging Codes},''
  \emph{arXiv preprint arXiv:2407.18393}, 2024. [Online]. Available:
  \url{https://arxiv.org/abs/2407.18393}
\BIBentrySTDinterwordspacing

\bibitem{xu2024_Fast_Parallel_LogicalComputation_Homological_Product}
Q.~Xu, H.~Zhou, G.~Zheng, D.~Bluvstein, J.~Ataides, M.~D. Lukin, and L.~Jiang,
  ``Fast and parallelizable logical computation with homological product
  codes,'' \emph{arXiv preprint arXiv:2407.18490}, 2024.

\bibitem{bluvstein2023Reconfig_Atom_Exp}
\BIBentryALTinterwordspacing
D.~Bluvstein, S.~J. Evered, A.~A. Geim \emph{et~al.}, ``Logical quantum
  processor based on reconfigurable atom arrays,'' \emph{Nature}, 2023.
  [Online]. Available: \url{https://doi.org/10.1038/s41586-023-06927-3}
\BIBentrySTDinterwordspacing

\bibitem{Matching_GB_NeutralAtoms_2024}
\BIBentryALTinterwordspacing
J.~Viszlai, W.~Yang, S.~F. Lin, J.~Liu, N.~Nottingham, J.~M. Baker, and F.~T.
  Chong, ``Matching generalized-bicycle codes to neutral atoms for low-overhead
  fault-tolerance,'' \emph{arXiv preprint arXiv:2311.16980}, 2024. [Online].
  Available: \url{https://arxiv.org/abs/2311.16980}
\BIBentrySTDinterwordspacing

\bibitem{dinur2023good}
I.~Dinur, M.-H. Hsieh, T.-C. Lin, and T.~Vidick, ``Good quantum {LDPC} codes
  with linear time decoders,'' in \emph{Proc. 55th Ann. ACM Symp. Theory of
  Computing}, 2023, pp. 905--918.

\bibitem{gu2024singleshotQLDPC}
S.~Gu, E.~Tang, L.~Caha, S.~H. Choe, Z.~He, and A.~Kubica, ``Single-shot
  decoding of good quantum {LDPC} codes,'' \emph{Communications in Mathematical
  Physics}, vol. 405, no.~3, p.~85, 2024.

\bibitem{panteleev_degenerate_2021}
\BIBentryALTinterwordspacing
P.~Panteleev and G.~Kalachev, ``Degenerate quantum {LDPC} codes with good
  finite length performance,'' \emph{Quantum}, vol.~5, p. 585, 2021. [Online].
  Available: \url{http://arxiv.org/abs/1904.02703}
\BIBentrySTDinterwordspacing

\bibitem{calderbank1996quantum_exists}
\BIBentryALTinterwordspacing
A.~R. Calderbank and P.~W. Shor, ``Good quantum error-correcting codes exist,''
  \emph{Phys. Rev. A}, vol.~54, pp. 1098--1105, Aug. 1996. [Online]. Available:
  \url{https://link.aps.org/doi/10.1103/PhysRevA.54.1098}
\BIBentrySTDinterwordspacing

\bibitem{Gottesman97}
\BIBentryALTinterwordspacing
D.~Gottesman, ``Stabilizer codes and quantum error correction,'' Dissertation,
  California Institute of Technology, 1997. [Online]. Available:
  \url{https://resolver.caltech.edu/CaltechETD:etd-07162004-113028}
\BIBentrySTDinterwordspacing

\bibitem{Brest_2023_NN_Quantum_Parallel}
A.~K. Pradhan, N.~Raveendran, N.~Rengaswamy, X.~Xiao, and B.~Vasi{\'c},
  ``Learning to decode trapping sets in {QLDPC} codes,'' in \emph{Proc. 12th
  Intl. Symp. Topics in Coding}, Sep. 2023, pp. 1--5.

\bibitem{Fossorier04}
M.~Fossorier, ``Quasicyclic low-density parity-check codes from circulant
  permutation matrices,'' \emph{IEEE Trans. on Inform. Theory}, vol.~50, no.~8,
  pp. 1788--1793, Aug. 2004.

\bibitem{Divsalar06}
D.~Divsalar, S.~Dolinar, and C.~Jones, ``Construction of protograph {LDPC}
  codes with linear minimum distance,'' in \emph{Proc. IEEE Intl. Symp. on
  Inform. Theory}, Seattle, WA USA, Jul. 2006, pp. 664--668.

\bibitem{Smarandache12}
R.~Smarandache and P.~O. Vontobel, ``{Quasi-Cyclic {LDPC} Codes: Influence of
  Proto- and Tanner-Graph Structure on Minimum {H}amming Distance Upper
  Bounds},'' \emph{IEEE Trans. on Inform. Theory}, vol.~58, no.~2, pp.
  585--607, Feb. 2012.

\bibitem{Tanner_QC}
R.~Tanner, D.~Sridhara, A.~Sridharan, T.~Fuja, and D.~Costello, ``{LDPC} block
  and convolutional codes based on circulant matrices,'' \emph{IEEE Trans. on
  Inform. Theory}, vol.~50, no.~12, pp. 2966--2984, 2004.

\bibitem{arrayCode_fan}
J.~L. Fan, ``Array codes as low-density parity-check codes,'' in \emph{Proc.
  2nd Intl. Symp. Turbo Codes and Related topics}, Sept. 2000, pp. 543--546.

\bibitem{MD99}
D.~J. MacKay and M.~C. Davey, ``Evaluation of {Gallager} codes for short block
  length and high rate applications,'' in \emph{In Codes, Systems and Graphical
  Models}.\hskip 1em plus 0.5em minus 0.4em\relax Springer-Verlag, 1999, pp.
  113--130.

\bibitem{Steane-physreva96}
A.~M. Steane, ``{Simple quantum error-correcting codes},'' \emph{Phys. Rev. A},
  vol.~54, no.~6, pp. 4741--4751, 1996.

\bibitem{Nielsen}
M.~A. Nielsen and I.~L. Chuang, \emph{Quantum Computation and Quantum
  Information: 10th Anniversary Edition}, 10th~ed.\hskip 1em plus 0.5em minus
  0.4em\relax New York, NY, USA: Cambridge University Press, 2011.

\bibitem{webster2024CodeConstruction_Evolutionary}
\BIBentryALTinterwordspacing
M.~Webster and D.~Browne, ``Engineering quantum error correction codes using
  evolutionary algorithms,'' \emph{arXiv preprint arXiv:2409.13017}, 2024.
  [Online]. Available: \url{https://arxiv.org/abs/2409.13017}
\BIBentrySTDinterwordspacing

\bibitem{Vardy97}
A.~Vardy, ``The intractability of computing the minimum distance of a code,''
  \emph{IEEE Trans. on Inform. Theory}, vol.~43, no.~6, pp. 1757--1766, Nov
  1997.

\bibitem{Pryadko_DistVerify_QLDPC}
I.~{Dumer}, A.~A. {Kovalev}, and L.~P. {Pryadko}, ``Distance verification for
  classical and quantum {LDPC} codes,'' \emph{IEEE Trans. on Inform. Theory},
  vol.~63, no.~7, pp. 4675--4686, 2017.

\bibitem{kapshikar_QminD_hardness_2023}
U.~Kapshikar and S.~Kundu, ``On the hardness of the minimum distance problem of
  quantum codes,'' \emph{{IEEE} Trans. on Inform. Theory}, pp. 1--1, 2023.

\bibitem{panteleev2022quantumAlmostLinearMinD}
P.~Panteleev and G.~Kalachev, ``Quantum {LDPC} codes with almost linear minimum
  distance,'' \emph{IEEE Trans. on Inform. Theory}, vol.~68, no.~1, pp.
  213--229, 2022.

\bibitem{raveendran2021trapping}
N.~Raveendran and B.~Vasi{\'c}, ``{Trapping sets of quantum LDPC codes},''
  \emph{Quantum}, vol.~5, p. 562, 2021.

\bibitem{imai_qc_codes}
M.~Hagiwara and H.~Imai, ``Quantum quasi-cyclic {LDPC} codes,'' in \emph{Proc.
  IEEE Intl. Symp. on Inform. Theory}, Jun. 2007, pp. 806--810.

\bibitem{galindo_quasi-cyclic_2018}
C.~Galindo, F.~Hernando, and R.~Matsumoto, ``Quasi-cyclic constructions of
  quantum codes,'' \emph{Finite Fields and Their Appl.}, vol.~52, pp. 261--280,
  2018.

\bibitem{Berrou02_1}
C.~Berrou, S.~Vaton, M.~Jezequel, and C.~Douillard, ``Computing the minimum
  distance of linear codes by the error impulse method,'' in \emph{Proc. IEEE
  Global Telecommun. Conf.}, vol.~2, Tapei, Taiwan, Nov. 2002, pp. 1017--1020.

\bibitem{Declercq2008}
D.~Declercq and M.~Fossorier, ``Improved impulse method to evaluate the low
  weight profile of sparse binary linear codes,'' in \emph{Proc. IEEE Int.
  Symp. on Inform. Theory}, Toronto, Canada, Jul, 2008, pp. 1963--1967.

\end{thebibliography}
\end{document}